\documentclass[11pt]{article}
\usepackage{graphicx} 
\usepackage{hyperref}
\usepackage{titletoc}
\usepackage{graphicx}
\usepackage{color} 
\usepackage{amsmath, amsfonts, amssymb, amsxtra, amsthm}
\usepackage{mathrsfs}
\usepackage{dsfont}
\usepackage{verbatim}
\usepackage{comment}
\usepackage{enumitem}
\usepackage{tikz}
\theoremstyle{definition}
\newtheorem{definition}{Definition}
\newtheorem{dn}{Definition}
\newtheorem{theorem}[dn]{Theorem}

\newtheorem{lemma}[dn]{Lemma}
\usepackage{mathtools}
\usepackage[noend,boxed]{algorithm2e}
\usepackage[noadjust]{cite}
\usepackage{placeins}
\usepackage{mdframed}

\DeclarePairedDelimiter\norm{\|}{\|}
\DeclarePairedDelimiter\abs{\lvert}{\rvert}
\newcommand{\R}{\mathbb{R}}
\newcommand{\E}{\mathbb{E}}
\newcommand{\pa}[1]{\left(#1\right)}
\newcommand{\Oish}{\widetilde{O}}

\usepackage[margin=1in]{geometry}
\usepackage{makecell}
\newcommand{\Szemeredi}{Szemer{\' e}di}

\newcommand{\eps}{\varepsilon}

\usepackage{todonotes}

\title{Improved Online Reachability Preservers\thanks{This work was supported by NSF:AF 2153680.}}
\author{}
\author{Greg Bodwin and Tuong Le\\University of Michigan\\\texttt{\{bodwin, tuongle\}@umich.edu}}
\date{}

\begin{document}

\maketitle

\thispagestyle{empty}

\begin{abstract}
A \emph{reachability preserver} is a basic kind of graph sparsifier, which preserves the reachability relation of an $n$-node directed input graph $G$ among a set of given demand pairs $P$ of size $|P|=p$.
We give constructions of sparse reachability preservers in the online setting, where $G$ is given on input, the demand pairs $(s, t) \in P$ arrive one at a time, and we must irrevocably add edges to a preserver $H$ to ensure reachability for the pair $(s, t)$ before we can see the next demand pair.
Our main results are:
\begin{itemize}
\item There is a construction that guarantees a maximum preserver size of
$$|E(H)| \le O\left( n^{0.72}p^{0.56} + n^{0.6}p^{0.7} + n\right).$$
This improves polynomially on the previous online upper bound of $O( \min\{np^{0.5}, n^{0.5}p\}) + n$, implicit in the work of Coppersmith and Elkin [SODA '05].

\item Given a promise that the demand pairs will satisfy $P \subseteq S \times V$ for some vertex set $S$ of size $|S|=:\sigma$, there is a construction that guarantees a maximum preserver size of
$$|E(H)| \le O\left( (np\sigma)^{1/2} + n\right).$$
A slightly different construction gives the same result for the setting $P \subseteq V \times S$.
This improves polynomially on the previous online upper bound of $O( \sigma n)$ (folklore).
\end{itemize}

All of these constructions are polynomial time, deterministic, and they do not require knowledge of the values of $p, \sigma$, or $S$.
Our techniques also give a small polynomial improvement in the current upper bounds for \emph{offline} reachability preservers, and our results extend to an even stronger model in which we must commit to a path for all possible reachable pairs in $G$ before any demand pairs have been received.
As an application, we improve the competitive ratio for Online Unweighted Directed Steiner Forest to $O(n^{3/5 + \eps})$, improving on the previous bound of $O(n^{2/3 + \eps})$ [Grigorescu, Lin, Quanrud APPROX-RANDOM '21].
\end{abstract}

\clearpage

\setcounter{page}{1}
\section{Introduction}

We study \emph{reachability preservers}, a basic graph sparsifier that has found applications in graph spanners and shortcut sets \cite{HP18}, property testing algorithms \cite{Alon02}, flow/cut approximation algorithms \cite{chuzhoy2009polynomial}, Steiner network design algorithms \cite{AB18}, etc (see \cite{BHT23} for more discussion and applications).

\begin{definition} [Reachability Preservers]
Given a directed graph $G = (V, E)$ and a set of demand pairs $P \subseteq V \times V$, a \emph{reachability preserver} is a subgraph $H \subseteq G$ with the property that, for all $(s, t) \in P$, there is an $s \leadsto t$ path in $H$ iff there is one in $G$.
\end{definition}

The study of reachability preservers goes back at least to Directed Steiner Network (DSN), a classic NP-hard graph algorithm, which can be phrased as the computational task of computing the sparsest (or minimum weight) reachability preserver of a given instance $G, P$ \cite{Winter87}.
More recently they have been intensively studied from an extremal perspective, where the goal is to determine the worst-case number of edges needed for a reachability preserver, typically as a function of $n$ (the number of nodes in the input graph) and $p := |P|$ (the number of demand pairs) \cite{AB18, BHT23, CC20, CCC22, Choudhary16, BCR16, GLQ21}.

Almost all previous work on reachability preservers operates in the \textbf{offline} model, where $G, P$ are given on input and the goal is to construct a sparse preserver $H$.
However, there is also a long line of algorithmic work on \textbf{online} Steiner network design \cite{GLQ21, ECKP15, BN09, AAABN06, AAA03, BN06, BN+09, GW95}.
Here the model is that $G$ is given on input, and the demand pairs $(s, t) \in P$ arrive one at a time.
We must irrevocably add edges to $H$ to preserve reachability for the current demand pair before the next one arrives, and the process can halt at any time without warning.
These papers typically try to achieve a small \emph{competitive ratio}, that is, the goal is design an efficient online algorithm that achieves an upper bound on preserver size of the form $|E(H)| \le OPT \cdot f(n, p)$ where $OPT$ is the best \emph{offline} solution for the given instance $G, P$ and the function $f$ is as small as possible.

We study extremal bounds for reachability preservers in the online model.
This problem has been previously explicitly considered only in a stronger non-standard online model (see Theorem \ref{thm:strongonline}), or implicitly as an internal ingredient inside the aforementioned algorithms (see Section \ref{sec:introudsn}).
There is one previous paper that implies results in this setting, which is by Coppersmith and Elkin \cite{CE06}, and it more strongly studies \emph{distance} preservers (which must preserve exact distances among demand pairs).
Although not explicitly discussed, one can use a folklore reduction of reachability preservers into the setting of DAGs (see Theorem \ref{thm:dagreduction}), and then apply the analysis of Coppersmith and Elkin to show an online upper bound of
$$|E(H)| \le O\left( \min\left\{np^{1/2}, n^{1/2}p\right\} + n \right).$$
Meanwhile, in the \emph{source-restricted} setting where the demand pairs satisfy $P \subseteq S \times V$ (or $P \subseteq V \times S$) for some set of source nodes $S$, we have the online bound
$$|E(H)| \le O\left( n|S| \right).$$
This construction is folklore; one simply selects paths for demand pairs in a ``consistent'' fashion \cite{CE06, Bodwin21}, meaning that they will all lie within a set of in- and out-trees rooted at the nodes in $S$, which have $O(n|S|)$ edges in total.
This simple construction remains state of the art, and has been used recently as an ingredient in online Steiner network algorithms \cite{GLQ21, ECKP15}, thus motivating further investigation.
We note that neither of these constructions require any advance knowledge of $P$ or $S$ (not even their size); this is typically considered to be an essential feature of the online model.

Although there are many more previously-known extremal bounds for offline reachability preservers (see Table \ref{tbl:priorwork}), they all rely on one or more technical tools that inherently require advance knowledge of the demand pairs; we discuss these tools in more detail in Section \ref{sec:onlinepairwise}.
We develop a framework in which to analyze online reachability preservers, and we use it to prove two new upper bounds, improving on both of the results mentioned above:

\begin{theorem} [Online Pairwise Reachability Preservers] \label{thm:intropairwise}
Given an $n$-node directed graph $G$, there is an online algorithm that constructs a reachability preserver $H$ of $|P|=p$ total demand pairs of size at most
\begin{equation*}
    |E(H)| \le O\pa{n^{\frac{2+\alpha}{3+\alpha}+o(1)}p^{\frac{2}{3+\alpha}} + n^{2-2\alpha + o(1)}p^{\alpha} + n} < O\left(n^{0.72} p^{0.56} + n^{0.6}p^{0.7} + n \right),
\end{equation*}
where $\alpha \ge 0.7$ is a root of $4 x^3- 13 x^2+10 x -2$.\footnote{This upper bound on $|E(H)|$ is decreasing as $\alpha$ increases, so one gets a correct but slightly suboptimal upper bound by plugging in $\alpha = 0.7$ (the explicit form on the right is a slight overestimate of that bound, with the exponents rounded off).  The exact value is $\alpha=0.70086\dots$.}
\end{theorem}

\begin{theorem} [Online Source-Restricted Reachability Preservers] \label{thm:introsourcewise}
Given an $n$-node directed graph $G = (V, E)$, and a promise that the demand pairs will satisfy $P \subseteq S \times V$ for some set of source nodes $S \subseteq V$, there is an online algorithm that constructs a reachability preserver $H$ of $p$ total demand pairs of size at most
$$|E(H)| \le O\left( (n|S|p)^{1/2} + n \right).$$
The same result holds under the promise that $P \subseteq V \times S$, but requires a slightly different algorithm.
\end{theorem}

All of our construction algorithms are deterministic, run in polynomial time, and do not require any knowledge of $P$ or $S$ (including their size).
Theorem \ref{thm:introsourcewise} ties the state-of-the-art upper bound in the \emph{offline} sourcewise setting \cite{AB18}.

It may also be interesting to compare Theorem \ref{thm:intropairwise} to the following lower bound, proved in \cite{BHT23} in a stronger online model:
\begin{theorem} [\cite{BHT23}] \label{thm:strongonline}
Consider a stronger version of the online model where the graph $G$ is initially empty, and an adversary may add new edges to $G$ in each round before providing the next demand pair.
Then there is a strategy for the adversary that guarantees that any online reachability preserver $H$ will have size
$$|E(H)| \ge \Omega\left( (np)^{2/3} + n\right).$$
\end{theorem}

Our bound in Theorem \ref{thm:intropairwise} is polynomially better than this one (in exchange for a weaker adversary), and so it formally separates these two models.

\renewcommand{\arraystretch}{1.5}

\begin{table}
\begin{tabular}{llccc}
\textbf{Bound on $|E(H)|$} & & \textbf{Offline} & \textbf{Online} & \textbf{Citation}\\
\hline
$O\big(\min\left\{np^{1/2}, n^{1/2}p\right\}$ & $+ n \big)$ & \checkmark{} & \checkmark{} & \cite{CE06} (implicit)\\
$O\big( n^{2/3} p^{2/3}$ & $+ n \big)$ & \checkmark{} & & \cite{AB18}\\
$O\big( \frac{n^2}{2^{\log^* n}}$ & $ + p \big)$ & \checkmark{} & & \cite{AB18}\\
$O\big( n^{3/4} p^{1/2} + n^{5/8} p^{11/16}$ & $+ n \big)$ & \checkmark{} & & \cite{BHT23}\\
$O\big(\frac{p^2}{2^{\log^* p}}$ & $+n \big)$ & \checkmark{} & & \cite{BHT23} \\
$O\big( n^{3/4} p^{1/2} + n^{2 - \sqrt{2} + o(1)} p^{1/\sqrt{2}}$ & $+ n\big)$ & \checkmark{} & & this paper\\
$O\big( n^{0.73}p^{0.54} + n^{0.6}p^{0.7}$ & $+ n\big)$ & \checkmark{} & \checkmark{} & this paper\\
\hline
\hline
$O\big( (np |S|)^{1/2} $ & $+n \big)$ & \checkmark{} & & \cite{AB18}\\
$O\big( (np |S|)^{1/2} $ & $+n \big)$ & \checkmark{} & \checkmark{} & this paper\\
\hline
\hline
For any int $d \ge 1$: $\Omega\big( n^{\frac{2}{d+1}} p^{\frac{d-1}{d}}$ & $+n+p \big)$ & \checkmark{} & \checkmark{} & \cite{AB18}, based on \cite{CE06}\\
\hline
\hline
For $p \ge n^{4/9} |S|^{2/3}$: $\Omega\big( n^{4/5} p^{1/5} |S|^{1/5}$ & $+n+p\big)$ & \checkmark{} & \checkmark{} & \cite{AB18}
\end{tabular}
\caption{\label{tbl:priorwork} The progression of upper and lower bounds on the extremal number of edges needed for an offline/online reachability preserver.  When $S$ is present, the bound applies in the source-restricted setting $P \subseteq S \times V$ (or $P \subseteq V \times S$).  The terms $2^{\log^* n}$ and $2^{\log^* p}$ reflect the current lower bounds on the Ruzsa-\Szemeredi{} function; see \cite{BHT23} for discussion.}
\end{table}

\subsection{Other Models}

A couple of the new tools that we develop for the online setting also turn out to be helpful in the classical offline setting.
This yields the following small polynomial improvement in the state of the art:

\begin{theorem} [Offline Reachability Preservers]\label{thm:offlinerp}
Given an $n$-node directed graph $G$ and a set of demand pairs $P$, one can construct in polynomial time a reachability preserver $H$ of size
$$|E(H)| \le O\left( n^{3/4} p^{1/2} + n^{2 - \sqrt{2} + o(1)} p^{1/\sqrt{2}} + n\right) \le O\left( n^{0.75} p^{0.5} + n^{0.59} p^{0.71} + n\right).$$
\end{theorem}

We will also consider the stronger \emph{non-adaptive} version of the online model.
In its strongest form, this model would require that the selected path $\pi(s, t)$ added to the preserver for each demand pair $(s, t)$ depends only on the input graph $G$, and not also on the previous demand pairs or the current preserver $H$.
Equivalently, after receiving the input graph $G$, we are required to commit to a choice of path for \emph{every} reachable pair $(s, t)$ before \emph{any} demand pairs are received.
In addition to the interpretation through the online model, this can be viewed as a parallelization of reachability preservers: it allows one to preprocess a graph $G$ in such a way that, given any set of demand pairs $P$, we can construct a sparse reachability preserver $H$ by adding a path for all demand pairs $(s, t) \in P$ \emph{in parallel}, without the path-adding process for $(s, t)$ even knowing the other demand pairs in $P$.

We will show that our upper bounds can \emph{almost} be made non-adaptive, but a little bit of extra information is required.
Making them fully non-adaptive is an interesting open problem.

\begin{theorem} [Non-Adaptive Reachability Preservers] \label{thm:introna}
There are online algorithms satisfying Theorem \ref{thm:intropairwise} and Theorem \ref{thm:introsourcewise} where the selected path for each demand pair $(s, t)$ depends only on the input graph $G$, the set of source nodes $S$ (for Theorem \ref{thm:introsourcewise}), and
\begin{itemize}
\item the index $i$ of the current demand pair, \textbf{or}
\item the total number $p$ of demand pairs that will arrive.
\end{itemize}
\end{theorem}

The first ($i$-sensitive) part of this theorem strictly strengthens Theorem \ref{thm:intropairwise}.
The second ($p$-sensitive) part is incomparable in strength to Theorem \ref{thm:intropairwise}, since it depends on the overall number of demand pairs $p$, which is unknown in Theorem \ref{thm:intropairwise}.
Both parts are incomparable in strength to Theorem \ref{thm:introsourcewise}, since they require knowledge of the source nodes $S$ which are unknown in Theorem \ref{thm:introsourcewise}.
These are proved in Section \ref{sec:nonadaptive}.
For comparison, the results of Coppersmith and Elkin for distance preservers \cite{CE06} imply non-adaptive reachability preservers of quality $O(\min\{np^{1/2}, n^{1/2}p\} + n)$, but nothing further was previously known.

\subsection{Application to Unweighted Directed Steiner Network \label{sec:introudsn}}

The algorithmic problem of computing the sparsest\footnote{More generally, in algorithmic contexts one can consider a weighted version of the problem, where $G$ has edge weights and the goal is to construct a min-weight preserver $H$.} possible reachability preserver of an input $G, P$ is called \emph{Unweighted Directed Steiner Network (UDSN)}.
This problem is NP-hard, and even hard to approximate \cite{DK99}, but it is well studied from the standpoint of approximation algorithms.
After considerable research effort \cite{CCCDGGL99, CEGS11, FKN12, BBMRY13, CDKL17, AB18}, the state-of-the-art is:\footnote{All work to date has focused on proving an approximation ratio as a function of $n$ \emph{or} as a function of $p$.  It is not clear if a better approximation ratio can be achieved if we consider bounds that can depend on both $n$ and $p$.}
\begin{theorem} [Offline UDSN \cite{CEGS11, AB18}] \label{thm:offlineudsn}
There is a randomized polynomial time algorithm for (offline) UDSN that, given an $n$-node input graph $G$ and a set of $p$ demand pairs $P$, returns a reachability preserver $H$ of size
$$|E(H)| \le OPT \cdot O\left(\min\left\{n^{4/7 + \eps}, p^{1/2 + \eps}\right\}\right)$$
where $OPT$ is the number of edges in the sparsest possible reachability preserver of $G, P$.
\end{theorem}
The bound of $O(n^{4/7 + \eps})$ follows a proof framework that was originally developed by Chlamt{\'a}{\v{c}}, Dinitz, Kortsarz, and Laekhanukit \cite{CDKL17}, but then replaces a certain key internal ingredient in their proof with an offline source-restricted reachability preserver \cite{AB18}.

It is an interesting open question whether the bound in Theorem \ref{thm:offlineudsn} can be matched by an online algorithm (in other words, is the \emph{competitive ratio} for UDSN bounded by $O\left(\min\left\{n^{4/7 + \eps}, p^{1/2 + \eps}\right\}\right)$?).
This was partially achieved about ten years ago by Chakrabarty, Ene, Krishnaswamy, and Panigrahi \cite{ECKP15}, who showed:
\begin{theorem} [Online UDSN, parametrized on $p$ \cite{ECKP15}]
There is a randomized polynomial time algorithm for online UDSN that achieves a competitive ratio of
$O\left( p^{1/2 + \eps} \right).$
\end{theorem}

Thus, the remaining question is whether the dependence on $n$ in Theorem \ref{thm:offlineudsn} can be recovered in the online setting.
Here the state of the art is due to, Grigorescu, Lin, and Quanrud, who obtained a competitive ratio of $n^{2/3 + \eps}$ \cite{GLQ21}, roughly following the framework of Chlamt{\'a}{\v{c}} et al.~\cite{CDKL17} but with adaptations for the online setting.
As an application of our new online source-restricted preservers, we are able to substitute them into this framework, improving the competitive ratio:

\begin{theorem} [Online UDSN, parametrized on $n$] \label{thm:onlineDSN}
There is a randomized polynomial time algorithm for online UDSN that achieves a competitive ratio of $O(n^{3/5 + \eps})$.
\end{theorem}

\subsection{Organization}

\begin{itemize}
\item Section \ref{sec:prelims} recaps some useful technical preliminaries, largely from \cite{BHT23}.
\item Section \ref{sec:pathgrowth} sets up the algorithms used to select paths for demand pairs as they arrive in the online setting, and proves their basic properties.
\item Section \ref{sec:onlinesource} proves Theorem \ref{thm:introsourcewise}.
\item Section \ref{sec:onlinepairwise} informally overviews Theorems \ref{thm:intropairwise} and \ref{thm:offlinerp}.
The formal proofs are more technical, so we give them in Appendices \ref{app:onlinerp} and \ref{app:offlinerp}, respectively.
\item Section \ref{sec:nonadaptive} proves Theorem \ref{thm:introna}.
\item Section \ref{sec:onlineapp} proves Theorem \ref{thm:onlineDSN}.
\end{itemize}
\section{Technical Preliminaries \label{sec:prelims}}

We will recap some definitions and results from prior work that will be useful in our arguments to follow.

\subsection{The DAG reduction}

While proving all of our upper bounds, it will be convenient to assume that the input graph is a DAG.
That this assumption is without loss of generality comes from the following standard reduction, which we will recap somewhat briefly here:
\begin{theorem} [DAG Reduction (folklore)] \label{thm:dagreduction}
If there is an algorithm that constructs a reachability preserver $H$ of size
$$|E(H)| \le f(n, p, \sigma)$$
for any $p$ demand pairs using $|S|=\sigma$ source (or sink) nodes in an $n$-node DAG, then there is an algorithm that constructs a reachability preserver $H$ of size
$$|E(H)| \le f(n, p, \sigma) + 2n$$
in arbitrary graphs.
\end{theorem}
\begin{proof}
Given an input graph $G$, compute a strongly connected component (SCC) decomposition.
For each component $C$ in the decomposition, choose an arbitrary vertex $v \in C$ and an in- and out-tree rooted at $v$ spanning $C$.
There are $2(|C|-1)$ edges in these two trees, and hence there are slightly less than $2n$ edges in total, across all trees.

We add all edges in these trees to our reachability preserver $H$ at their first opportunity, and then for the rest of the construction we treat each SCC as a single contracted supernode, yielding a DAG $G'$.
For each demand pair $(s, t)$, we can map the nodes $s, t$ onto the corresponding supernodes in $G'$ and choose paths in $G'$ using our assumed DAG algorithm, to get a reachability preserver $H'$ in $G'$ on $\le f(n, p, \sigma)$ edges.
Each edge $(u, v)$ added to $H'$ can be mapped back to any single edge in $G$ between the set of nodes corresponding to the supernode $u$ and the set of nodes corresponding to the supernode $v$.
Since our in- and out-trees preserve strong among all nodes in each of these sets, this will give a correct preserver in $G$.
\end{proof}

\subsection{Path System Definitions}

A \emph{path system} is a pair $S = (V, \Pi)$, where $V$ is a set of vertices and $\Pi$ is a set of nonempty vertex sequences called \emph{paths}.
We will next recap some basic definitions; see also \cite{BHT23} for more discussion.
Note that, even when the vertex set $V$ is that of some graph $G$, the paths in $\Pi$ are abstract sequences of vertices that do not necessarily correspond to paths in $G$.
The \emph{length} of a path $\pi \in \Pi$, written $|\pi|$, is its number of vertices (hence off by one from the length of the path through some graph).
The \emph{degree} of a node $v \in V$, written $\deg(v)$, is the number of paths that contain $v$ (which may differ significantly from its degree as a node in a graph).
The \emph{size} of $S$, written $\|S\|$, is the quantity
$$\|S\| := \sum \limits_{\pi \in \Pi} |\pi|.$$

\usetikzlibrary{calc}
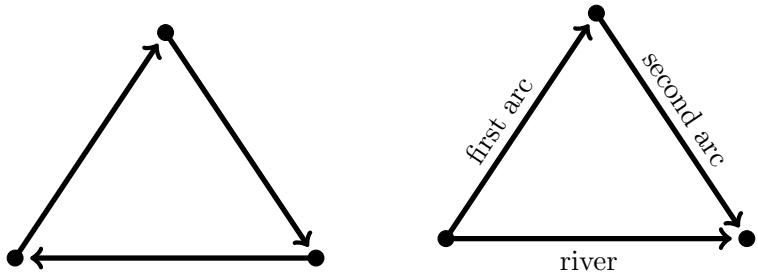
\begin{figure}[h]
\begin{center}

\begin{tikzpicture}

    \coordinate (A) at (0,0);
    \coordinate (B) at (4,0);
    \coordinate (C) at (2,3);

    \draw [fill=black] (A) circle (3pt);
    \draw [fill=black] (B) circle (3pt);
    \draw [fill=black] (C) circle (3pt);

    \draw [->, line width = 2pt] (B) -- ($(B)!.95!(A)$) node[midway, below, sloped] {};
    \draw [->, line width = 2pt] (C) -- ($(C)!.95!(B)$) node[midway, above, sloped] {};
    \draw [->, line width = 2pt] (A) -- ($(A)!.95!(C)$) node[midway, above, sloped] {};
\end{tikzpicture}\hspace{1.5cm}%
\begin{tikzpicture}

    \coordinate (A) at (0,0);
    \coordinate (B) at (4,0);
    \coordinate (C) at (2,3);

    \draw [fill=black] (A) circle (3pt);
    \draw [fill=black] (B) circle (3pt);
    \draw [fill=black] (C) circle (3pt);

    \draw [->, line width = 2pt] (A) -- ($(A)!.95!(B)$) node[midway, below, sloped] {river};
    \draw [->, line width = 2pt] (C) -- ($(C)!.95!(B)$) node[midway, above, sloped] {second arc};
    \draw [->, line width = 2pt] (A) -- ($(A)!.95!(C)$) node[midway, above, sloped] {first arc};

\end{tikzpicture}
\end{center}
\caption{\label{fig:subsystems} A directed $3$-cycle (left) and a $3$-bridge (right).}
\end{figure}

We will later use so-called \emph{Tur\'{a}n-type methods} to bound the size of certain path systems, meaning that we will first establish that $S$ avoids certain subsystems, and then we will bound the maximum possible size of \emph{any} path system that avoids those subsystems.
We say that $S'$ is a \emph{subsystem} of $S$, written $S' \subseteq S$, if it can be obtained by zero or more of the following operations: delete a path from $\Pi$, delete a node from $V$, or delete one instance of a node $v$ from a path $\pi \in \Pi$.
We will use two kinds of forbidden subsystems in this paper (see Figure \ref{fig:subsystems}):
\begin{itemize}
\item A \emph{directed $k$-cycle} is a path system that has $k$ nodes with a circular ordering $(x_0, x_1, \dots, x_{k-1}, x_k=x_0)$, and $k$ paths of length two each, which are all paths of the form $(x_i, x_{i+1}), 0 \le i < k$.

\item A \emph{$k$-bridge} is a path system that has $k$ nodes with a total ordering $(x_1, x_2, \dots, x_k)$, and $k$ paths of length two each, which are (1) the path $(x_1, x_k)$, called the \emph{river}, and (2) the $k-1$ paths of the form $(x_i, x_{i+1}), 1 \le i < k$, called the \emph{$i^{th}$ arc}.
\end{itemize}

We note that $k$-bridges still count even when they are degenerate; for example, two paths that coincide on two consecutive nodes count as a $2$-bridge, and three paths that coincide on three consecutive nodes contain a $3$-bridge, etc.

A path system is \emph{acyclic} if it does not contain any directed cycle as a subsystem, or equivalently, if there is a total ordering of the vertices $V$ such that the order of vertices within each path $\pi \in \Pi$ agrees with this ordering.
We will frequently consider \emph{ordered path systems}, which are path systems with a total ordering on their path set $\Pi$.
With this we will sometimes only forbid subsystems with certain ordering constraints, e.g., $3$-bridges where the last arc comes before the river in the ordering of $\Pi$.
\section{Online Reachability Preservers}

\subsection{Path Growth Algorithms \label{sec:pathgrowth}}

A key tool in our online upper bounds will be the following two (very similar) path selection algorithms.
We use these algorithms to generate a path $\pi(s, t)$ for each demand pair $(s, t)$ as it arrives, and then we add the edges of this path to the current preserver.
Our path generation algorithms are greedy, growing paths one edge at a time and locally avoiding new edges if possible.
The two algorithms are symmetric to each other, and differ only in whether we grow the path from front to back or from back to front.

\DontPrintSemicolon
\begin{algorithm}[t]
\textbf{Input:} DAG $G = (V, E)$, current preserver $H \subseteq G$, demand pair $(s, t)$\;~\\

Let $\pi \gets (s)$\;
\While{last node of $\pi$ is not $t$}{
    $u \gets $ last node of $\pi$\;
    \If{there exists an edge $(u, v) \in E(H)$ with $t$ reachable from $v$}{
        append any such edge $(u, v)$ to the back of $\pi$\;
    }
    \Else{
        append to the back of $\pi$ any edge $(u, v) \in E(G)$ with $t$ reachable from $v$\;
    }
}
\textbf{return} $\pi$\;

\caption{\texttt{forwards-growth} path generation}
\end{algorithm}

\begin{algorithm}[t]
\textbf{Input:} DAG $G = (V, E)$, current preserver $H \subseteq G$, demand pair $(s, t)$\;~\\

Let $\pi \gets (t)$\;
\While{first node of $\pi$ is not $s$}{
    $v \gets $ first node of $\pi$\;
    \If{there exists an edge $(u, v) \in E(H)$ with $u$ reachable from $s$}{
        append any such edge $(u, v)$ to the front of $\pi$\;
    }
    \Else{
        append to the front of $\pi$ any edge of the form $(u, v) \in E(G)$ with $u$ reachable from $s$\;
    }
}
\textbf{return} $\pi$\;

\caption{\texttt{backwards-growth} path generation}
\end{algorithm}

As we use these algorithms to sequentially generate paths and build our preserver, it will be helpful to track an auxiliary path system $Z = (V, \Pi)$.
Each time we add a path $\pi(s, t)$ to $H$, say that a \emph{new edge} is an edge $e \in \pi(s, t)$ that was not previously in the preserver.
We then add a corresponding path $\pi'$ to $Z$, whose nodes are
$$\pi' := \begin{cases}
\left\{ u \ \mid \ \text{there is a new edge } (u, v) \in \pi(s, t)\right\} \cup \{t\} \text{ if } \texttt{forwards-growth} \text{ is used}\\
\{s\} \cup \left\{ v \ \mid \ \text{there is a new edge } (u, v) \in \pi(s, t)\right\} \text{ if } \texttt{backwards-growth} \text{ is used}
\end{cases}$$
and in either case, these nodes are ordered in the path $\pi'$ the same as their order in $\pi(s,t)$.
We will also treat $Z$ as an ordered path system, with the paths in $Z$ ordered by the arrival of the demand pairs that generated each path.
The following properties of $Z$ all follow straightforwardly from the construction:\\

\begin{lemma} [Properties of $Z$] \label{lem:zprops} ~
\begin{enumerate}
\item $Z$ is acyclic,
\item $\|Z\| = |E(H)| + p$,
\item Under \texttt{forwards-growth}, $Z$ has no bridge in which the first arc comes before the river in the ordering of $\Pi$.
Under \texttt{backwards-growth}, $Z$ has no bridge in which the last arc comes before the river in the ordering of $\Pi$.
\end{enumerate} ~
\end{lemma}
\begin{proof} ~
\begin{enumerate}
\item Since the input graph $G = (V, E)$ is a DAG, the order of nodes in each path $\pi \in \Pi$ agrees with the topological ordering of the nodes in $V$, implying that $Z$ is acyclic.

\item Initially, we have $\|Z\| = |E(H)| = 0$.
Then, every path $\pi'$ added to $Z$ corresponds to a path $\pi(s, t)$ that contributes exactly $|\pi'|-1$ new edges to $H$, so in the end we have $\|Z\| = |E(H)| + p$.

\item We will prove this for \texttt{forwards-growth}; the argument for \texttt{backwards-growth} is symmetric (up to reversal of direction of the edges of the input graph $G$).
Seeking contradiction, suppose there is a bridge formed by nodes $(x_1, \dots, x_k)$, arc paths $\pi_1, \dots, \pi_{k-1}$, and river path $\pi_r$, with $\pi_1 <_{\Pi} \pi_r$.
Let $\pi(s_1, t_1), \pi(s_r, t_r)$ be the paths generated by \texttt{forwards-growth} corresponding to $\pi_1, \pi_r$ respectively.
By construction, since $x_1 \in (\pi_1 \cap \pi_r)$, these paths both contribute new edges to $H$ leaving $x_1$; call the first one $(x_1, y) \in \pi(s_1, t_1)$.
Now notice that the arcs witness reachability among all of the node pairs
$$\underbrace{(y, x_2)}_{\text{in } \pi_1}, \underbrace{(x_2, x_3)}_{\text{in } \pi_2}, \dots, \underbrace{(x_{k-1}, x_k)}_{\text{in } \pi_{k-1}}, \underbrace{(x_k, t_r)}_{\text{in } \pi_k}.$$
So by transitivity, the node pair $(y, t_r)$ is reachable.
When we generate $\pi(s_r, t_r)$ using \texttt{forwards-growth}, since we have already added $\pi(s_1, t_1)$ the edge $(x_1, y)$ is already present in $H$ and we have $(y, t_r)$ reachability.
So the algorithm will \emph{not} choose to add a new edge leaving $x_1$ while generating $\pi(s_r, t_r)$, which completes the contradiction. \qedhere
\end{enumerate}
\end{proof}

\subsection{Online Source-Restricted Reachability Preservers \label{sec:onlinesource}}

We will prove the following upper bound on source-restricted preservers in the online model:

\begin{theorem} \label{thm:onlinesource}
In the online model with an $n$-node input DAG $G = (V, E)$ and $p$ total demand pairs, the final preserver $H$ will have size
$$|E(H)| \le O\left((np|S|)^{1/2} + n \right)$$
in either of the following two settings:
\begin{itemize}
\item $S$ is the set of start nodes used by the given demand pairs $P$ (that is, $P \subseteq S \times V$), and the builder generates paths in each round using the \texttt{backwards-growth} algorithm, or

\item $S$ is the set of end nodes used by the given demand pairs $P$ (that is, $P \subseteq V \times S$), and the builder generates paths in each round using the \texttt{forwards-growth} algorithm.
\end{itemize}
\end{theorem}

Notably, neither the \texttt{forwards-} nor \texttt{backwards-}growth algorithm require the builder to know the number of demand pairs $p$ or any information about the set of source/sink nodes $S$.
We will only prove the latter point in Theorem \ref{thm:onlinesource}, analyzing \texttt{forwards-growth} and assuming $P \subseteq V \times S$.
The other point is symmetric.

The proof will work by analyzing the path system $Z$ associated to the online path-adding process; recall its essential properties in Lemma \ref{lem:zprops}.
Let $\ell := \|Z\|/p$ be the average path length and let $d := \|Z\|/n$ be the average node degree.
If $d \le O(1)$ then we have $\|Z\| \le O(n)$ and the theorem holds, so we may assume in the following that $d$ is at least a sufficiently large constant.
Imagine that we add the paths from $Z$ to an initially-empty system, one at a time, in the \textbf{reverse} of their ordering in $\Pi$.
We observe:
\begin{lemma}
There is a path $\pi \in \Pi$ such that when $\pi$ is added in the above process, it contains at least $\ell/4$ nodes of degree at least $d/4$ each.
\end{lemma}
\begin{proof}
Suppose not.
Then, by counting the first $d/4$ times each node appears in a path separate from the remaining times, the total size of $Z$ can be bounded as
\begin{align*}
\|Z\| &\le \frac{nd}{4} + \frac{p\ell}{4}\\
&\le \frac{\|Z\|}{4} + \frac{\|Z\|}{4}\\
&= \frac{\|Z\|}{2},
\end{align*}
which is a contradiction.
\end{proof}

\begin{figure}[h]
\begin{center}
\begin{tikzpicture}

    \coordinate (A) at (0,0);
    \coordinate (B) at (4,0);
    \coordinate (C) at (2,3);

    \draw [fill=black] (A) circle (3pt);
    \draw [fill=black] (B) circle (3pt);
    \draw [fill=black] (C) circle (3pt) node [above=0.2cm] {$t \in S$};

    \draw [->, line width = 2pt, dotted] (A) -- ($(A)!.95!(B)$) node[midway, below, sloped, align=center] {$\pi$\\(first in ordering, first arc)};
    \draw [->, line width = 2pt] (B) -- ($(B)!.95!(C)$) node[midway, above, sloped] {$q_2$ (second arc)};
    \draw [->, line width = 2pt] (A) -- ($(A)!.95!(C)$) node[midway, above, sloped] {$q_1$ (river)};

\end{tikzpicture}
\end{center}
\caption{\label{fig:source3bridge} The proof of Lemma \ref{lem:pathintbound} works by arguing that no path $\pi$ may intersect too many paths that come later in the ordering, or else two of those paths $q_1, q_2$ will share an endpoint in $S$ and thus form a forbidden $3$-bridge.}
\end{figure}
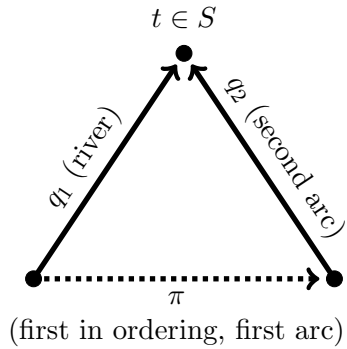

\begin{lemma} \label{lem:pathintbound}
$\ell d \le O(|S|).$
\end{lemma}
\begin{proof}
Suppose for contradiction that $\ell d > 16|S|$.
By the previous lemma, there is a path $\pi \in \Pi$ that intersects at least $\ell d / 16 > |S|$ other paths in $\Pi$, which were added to the system before $\pi$ (and hence come \emph{later} than $\pi$ in the ordering of $\Pi$).
By the Pigeonhole principle, and since the demand pairs satisfy $P \subseteq V \times S$, at least two of these intersecting paths $q_1, q_2$ end at the same node $t \in S$.
Since by Lemma \ref{lem:zprops} $Z$ does not contain any $2$-bridges, $q_1, q_2$ may not intersect at any other nodes, and so they intersect $\pi$ at two different nodes.
But this implies that $\pi, q_1, q_2$ form a $3$-bridge in which $\pi$ is the first arc \emph{and} it precedes both $q_1, q_2$ in the ordering of $\Pi$ (see Figure \ref{fig:source3bridge}).
This contradicts Lemma \ref{lem:zprops}, completing the proof.
\end{proof}

We now complete the proof by algebraically rearranging the inequality from the previous lemma.
We have:
\begin{align*}
\ell d &\le O(|S|)\\
(p\ell) (nd) &\le O(|S| pn)\\
\|Z\|^2 &\le O(|S| pn)\\
\|Z\| &\le O(|S| pn)^{1/2}.
\end{align*}

Since by Lemma \ref{lem:zprops} we have $\|Z\| \ge |E(H)|$, this implies our desired bound on the size of the output preserver.

\subsection{Online Pairwise Reachability Preservers \label{sec:onlinepairwise}}

By following an identical proof strategy to our upper bound in the source-restricted setting (i.e., exploiting forbidden ordered $2$- and $3$-bridges), it is possible to prove an upper bound of
$$|E(H)| \le O\left((np)^{2/3} + n\right)$$
(details omitted, since we will show a stronger bound than this).
As discussed in \cite{BHT23}, this is probably the best upper bound one can show by exploiting \emph{only} the forbidden ordered $2$- and $3$-bridges from Lemma \ref{lem:zprops}.
Nonetheless, we will show a stronger bound, which crucially also exploits the forbidden ordered $4$-bridges from Lemma \ref{lem:zprops}.

\paragraph{Recap of \cite{BHT23}.}

Recent work of Bodwin, Hoppenworth, and Trabelsi \cite{BHT23} on offline reachability preservers introduced a framework for extremal analysis of forbidden $4$-bridges, which we will briefly recap here.
First, the paper shows:
\begin{lemma} [Independence Lemma, c.f.~\cite{BHT23}, Lemma 38]
Let $\beta(n, p, \infty)$ denote the maximum possible size of a path system with $n$ nodes, $p$ paths, and no bridges as subsystems.
Then every $n$-node directed graph and set of $p$ demand pairs has an offline reachability preserver $H$ of size
$$|E(H)| \le O\left(\beta(n, p, \infty)\right),$$
and this is asymptotically tight.
\end{lemma}

Thus, it suffices to argue about the extremal size of a bridge-free path system.
We remark here that versions of this lemma are perhaps implicit at a low level in prior work, e.g.~\cite{CE06, AB18}.
It is also an inherently offline lemma, and breaks down completely in the online setting; the reliance on this lemma (or its underlying ideas) is essentially why the known results for offline reachability preservers do not tend to extend to the online setting.

This previous paper then argues as follows.
Assume for convenience that all paths have length $\Theta(\ell)$ and all nodes have degree $\Theta(d)$, for some parameters $\ell, d$.
Recall that the upper bound from forbidden $2$- and $3$-bridges is $O((np)^{2/3} + n)$, and so our goal is to show that this bound cannot be tight.
A few straightforward calculations reveal that, \emph{if} this bound were tight, then for the typical pair of paths $\pi_1, \pi_2$ in the system there will be $\Theta(\ell)$ paths that intersect $\pi_1$ and then $\pi_2$.
However, no pair of these intersecting paths may have crossing intersection points with $\pi_1, \pi_2$ (see Figure \ref{fig:nocrossing}).

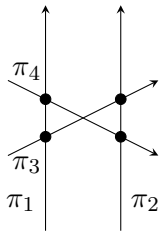
\begin{figure} [h]
\begin{center}
\begin{tikzpicture}[>=stealth] 
    \draw[->] (1.5,0) -- (1.5,3) node[very near start,left] {$\pi_1$};
    \draw[->] (2.5,0) -- (2.5,3) node[very near start,right] {$\pi_2$};

    \draw[->] (1,1) -- (3,2) node[very near start,below] {$\pi_3$};
    \draw[->] (1,2) -- (3,1) node[very near start,above] {$\pi_4$};

    \filldraw [black] (1.5,1.25) circle (2pt); 
    \filldraw [black] (1.5,1.75) circle (2pt); 

    \filldraw [black] (2.5,1.25) circle (2pt); 
    \filldraw [black] (2.5,1.75) circle (2pt); 

\end{tikzpicture}
\end{center}
\caption{\label{fig:nocrossing} There cannot be two paths $\pi_3, \pi_4$ that both intersect paths $\pi_1, \pi_2$, but where the points of intersection switch places as in this picture, or else they form a $4$-bridge (here $\pi_3$ is the river).}
\end{figure}

If there are $\Theta(\ell)$ paths that intersect both $\pi_1, \pi_2$, and yet these intersecting paths cannot cross each other, then the typical intersecting path must ``lie flat'' in the sense that there is not much of a gap along either $\pi_1$ or $\pi_2$ to, say, the $h$ nearest intersecting paths (for some parameter $h$).
In order to exploit this, the key strategy in \cite{BHT23} is to sample a random \emph{base path} $\pi_b \in \Pi$, and then analyze the random subsystem $S'$ on the vertex set formed by examining the $h$ adjacent nodes along the paths that intersect $\pi_b$ (see Figure \ref{fig:sprime}).

\usetikzlibrary{decorations.pathreplacing}

\begin{figure} [ht]
\begin{center}
    \begin{tikzpicture}
    \draw [fill=black] (0, 0) circle [radius=0.15];
    \draw [fill=black] (2, 0) circle [radius=0.15];
    \draw [fill=black] (4, 0) circle [radius=0.15];
    \draw [fill=black] (6, 0) circle [radius=0.15];
    \draw [ultra thick, ->] (-0.5, 0) -- (6.5, 0);
    \draw [ultra thick, ->] (-0.125, -0.5) -- (0.5, 2);
    \draw [ultra thick, ->] (0.125, -0.5) -- (-0.5, 2);
    
    \draw [ultra thick, ->] (1.875, -0.5) -- (2.5, 2);
    \draw [ultra thick, ->] (2.125, -0.5) -- (1.5, 2);
    
    \draw [ultra thick, ->] (3.875, -0.5) -- (4.5, 2);
    \draw [ultra thick, ->] (4.125, -0.5) -- (3.5, 2);
    
    \draw [ultra thick, ->] (5.875, -0.5) -- (6.5, 2);
    \draw [ultra thick, ->] (6.125, -0.5) -- (5.5, 2);
    
    \draw [decorate, decoration = {brace}] (-1, 0) --  (-1, 2);
    \node at (-2.5, 1) {$h$ nodes/path};
    \node at (7, 0) {$\pi_b$};
    
    
    
    \end{tikzpicture}
\end{center}
\caption{\label{fig:sprime} The random subsystem $S'$.  Figure based on \cite{BHT23}, Figure 7.}
\end{figure}
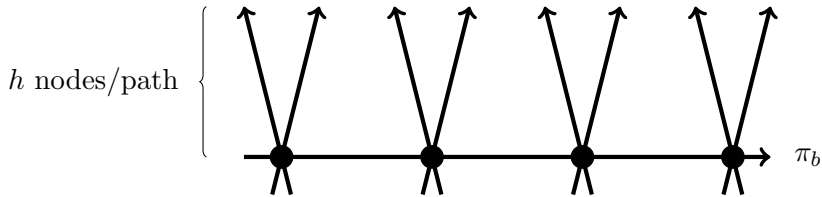

If the intersecting paths do indeed ``lie flat,'' then there will be many such paths within $h$ steps of $\pi_b$ along its branching paths, and thus we should expect $S'$ to contain many long paths.
But we can apply known upper bounds to $S'$ to rule out this possibility.
This implies that, in fact, the typical pair of paths $\pi_1, \pi_2$ have $\ll \ell$ paths that intersect both, leading to an improved upper bound.

\paragraph{Our Offline Improvements.}

An auxiliary result of this paper is an improvement in the bound shown by \cite{BHT23}.
We refer back to Theorem \ref{thm:offlinerp} for the statement, or Appendix \ref{app:offlinerp} for the proof.

The source of these improvements is from an improved strategy for controlling the size of the random subsystem $S'$.
One of the two ways in which this part improved is by \emph{recursively} bounding the size of $S'$ (this is executed in Lemma \ref{lem:l2boundbf}).
Although this idea is conceptually straightforward, it requires a significant refactoring of the proof to enable it.
The technical reason is that \cite{BHT23} bounds the input system $Z$ using an $\ell^1$ norm of path lengths (the standard notion of size) but $S'$ using an $\ell^2$ norm of path lengths, making it impossible to directly apply the bound on $Z$ recursively to $S'$.
We switch to bounding both using the $\ell^2$ norm everywhere, and we only move back to our desired $\ell^1$ norm at the very end of the proof.
The other new ingredient is an improved counting of the contribution of ``short'' paths to the size of $S'$, which is executed in Lemma \ref{lem:l2boundbf}. 

\paragraph{Our Online Adaptation.}

It will be slightly more convenient in this exposition to consider the \texttt{backwards-growth} strategy for path generation here, although of course by symmetry either strategy works (we use this convention in Appendix \ref{app:onlinerp} as well).
We will analyze the path system $Z$ constructed above, and in particular Lemma \ref{lem:zprops} states that bridges are forbidden in $Z$ whose \emph{last} arc comes before the river.

Can we exploit forbidden ordered $4$-bridges by following the strategy outlined above?
Some parts of the method extend easily, with minor tweaks.
For example, instead of counting \emph{any} paths $\pi$ that intersect the typical pair of paths $\pi_1, \pi_2$, we can restrict attention to those that also come after $\pi_1, \pi_2$ in the ordering.
Then the ``crossing'' in Figure \ref{fig:nocrossing} is still forbidden (since the river $\pi_3$ is assumed to come after the last arc $\pi_2$).
Relatedly, when we define $S'$, we need to consider only the paths that intersect $\pi_b$ \emph{and} come before it.
Nonetheless, all these definitional adaptations turn out to affect the relevant counting arguments by only constant factors, and so they do not harm the argument.
There are many other minor issues that we will not overview, which can be dispatched with a little technical effort.

However, there is one major problem: there is now potential for overlap in the paths that branch off $\pi_b$.
That is, in the offline setting, we can exploit $3$-bridge-freeness to argue that no two paths that branch off $\pi_b$ also intersect each other, and thus every node in $S'$ (except those on $\pi_b$ itself) is in exactly one branching path.
But in the online setting, this is not so: these intersections form a $3$-bridge that may be allowed (see Figure \ref{fig:sprimeorder}).

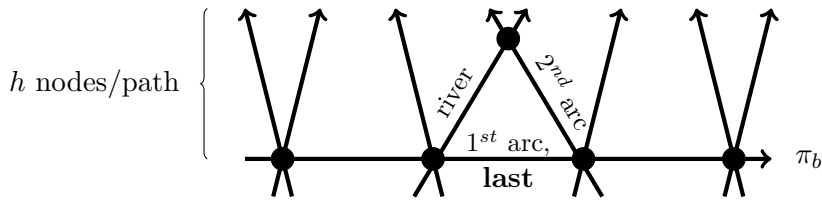
\begin{figure} [ht]
\begin{center}
    \begin{tikzpicture}
    \draw [fill=black] (0, 0) circle [radius=0.15];
    \draw [fill=black] (2, 0) circle [radius=0.15];
    \draw [fill=black] (4, 0) circle [radius=0.15];
    \draw [fill=black] (6, 0) circle [radius=0.15];
    \draw [ultra thick, ->] (-0.5, 0) -- (6.5, 0) node [midway, align=center] {\small $1^{st}$ arc,\\\bf last};
    \draw [ultra thick, ->] (-0.125, -0.5) -- (0.5, 2);
    \draw [ultra thick, ->] (0.125, -0.5) -- (-0.5, 2);
    
    \draw [ultra thick, ->] (1.75, -0.5) -- (3.25, 2) node [midway, sloped, above] {\small river};
    \draw [ultra thick, ->] (2.125, -0.5) -- (1.5, 2);

    \draw [fill=black] (3, 1.6) circle [radius=0.15];
    
    \draw [ultra thick, ->] (3.875, -0.5) -- (4.5, 2);
    \draw [ultra thick, ->] (4.25, -0.5) -- (2.75, 2) node [midway, sloped, above] {\small $2^{nd}$ arc};
    
    \draw [ultra thick, ->] (5.875, -0.5) -- (6.5, 2);
    \draw [ultra thick, ->] (6.125, -0.5) -- (5.5, 2);
    
    \draw [decorate, decoration = {brace}] (-1, 0) --  (-1, 2);
    \node at (-2.5, 1) {$h$ nodes/path};
    \node at (7, 0) {$\pi_b$};
    
    
    
    \end{tikzpicture}
\end{center}
\caption{\label{fig:sprimeorder} When we generate $S'$ in the online/ordered setting, it is possible for the paths intersecting $\pi_b$ to intersect each other: the river \emph{could} come before the $2^{nd}$ arc, which would not violate the conditions of Lemma \ref{lem:zprops}.}
\end{figure}

Naively, the typical node in $S'$ could be $\Theta(d)$ branching paths, and this complication completely wipes out all the gains from the analysis.
Although we cannot rule out the possibility of some nodes having $\Theta(d)$ branching paths, we are able to use a more intricate maneuver to show that there is some threshold $t \ll d$ such that, in expectation, the number of nodes of branching degree $\Omega(x)$ in $S'$ is far less than the trivial bound of $\ell d h / x$ for $x \ge t$.
This turns out to be good enough to recover some gains from the method.
We unfortunately cannot show this for $t=1$, which is the fundamental reason why our online bounds are polynomially worse than the corresponding offline bounds.
\section{Non-Adaptive Reachability Preservers \label{sec:nonadaptive}}

We next describe our method to convert our online algorithms to (almost) non-adaptive algorithms.
Our algorithm essentially works by simulating a \emph{greedy adversary} in the online model, who repeatedly provides the most costly demand pair in each round, and then we set paths by running our online algorithm against this adversary, halting at the appropriate place.
Our proof will thus imply indirectly that this greedy strategy is the most effective one for an adversary in the online model, up to constant factors.

Our algorithms will reference an \emph{extremal function} $f(n, p)$ for online reachability preservers, achieved by a generic path selection algorithm $\pi(s, t \mid G, H)$.
That is, we imagine an algorithm that starts with $G$ as any $n$-node graph, and $H$ as the $n$-node empty graph.
When each demand pair $(s, t)$ arrives, we select the path $\pi(s, t \mid G, H)$ and add all of its edges to $H$.
Then $f(n, p)$ is the largest possible number of edges in $H$ after $p$ rounds of this process.

We use this generic extremal function $f$, rather than the particular extremal upper bound from Theorem \ref{thm:intropairwise}, in order to emphasize that this bound and the technical details of the \texttt{forwards-} or \texttt{backwards-growth} algorithms are not really important in this proof.
If a future result improves on Theorem \ref{thm:intropairwise}, then the new bound will automatically transfer to these results as well.
For simplicity we will focus on $f(n, p)$ here, but the proof would generalize readily to extremal functions that incorporate additional parameters beyond $n$ and $p$.
This includes the online source-restricted preservers of Theorem \ref{thm:onlinesource}, which incorporate $|S|$ as a parameter, although we note that these require the set $S$ to be given on input, so that the we know the set of possible demand pairs and we can properly simulate the adversary (i.e., search over the proper subset of demand pairs in the condition of the while loop).

\begin{algorithm}
\textbf{Input:} $n$-node directed graph $G = (V, E)$, number of demand pairs $p$\\

All reachable node pairs $(s, t)$ in $G$ have ``unfinalized'' path\;
$H \gets (V, \emptyset)$\;
Let $f(n, p)$ be an extremal function for online reachability preservers, achieved by a deterministic path selection algorithm $\pi(s, t \mid G, H)$\;

\While{there is reachable $(s, t)$ with $>f(n, p)/p$ edges in $\pi(s, t \mid G, H) \setminus E(H)$}{
    finalize path $\pi(s, t \mid G, H)$ for $(s, t)$\;
    add edges of $\pi(s, t \mid G, H)$ to $H$\;
}

\ForEach{remaining unfinalized reachable pair $(s, t)$}{
    finalize path $\pi(s, t \mid G, H)$ for $(s, t)$\;
}

\caption{\label{alg:knownpnonadaptive} \texttt{known-p-non-adaptive-rps}}
\end{algorithm}

\begin{theorem} \label{thm:knownpnonadaptive}
For all $n$-node graphs $G$ and sequences $P$ of $|P| =: p$ demand pairs, the paths set by Algorithm \ref{alg:knownpnonadaptive} satisfy
$$\left|\bigcup \limits_{(s, t) \in P} \pi(s, t) \right| \le 2f(n, p).$$
\end{theorem}
\begin{proof}
Let $Q$ be the set of demand pairs whose paths are set in the initial while loop, and let $H_Q$ be the subgraph $H$ just after the paths for $Q$ have been set and the while loop terminates.
We first note that $|Q| < p$, since otherwise by counting the edges contributed to $H$, the first $p$ demand pairs in $Q$ create a subgraph $H_Q$ of size $|E(H_Q)| > f(n, p)$ which contradicts the definition of the extremal function $f$.
Since $|Q| < p$, we therefore have
$$\left|E(H_Q)\right| \le f(n, p).$$
Meanwhile, all demand pairs in $P \setminus Q$ have their path set in the final for loop, and by construction there are $\le f(n, p) / p$ edges outside $H_Q$ in each path.
So we have
\begin{align*}
\left|\bigcup \limits_{(s, t) \in P} \pi(s, t) \right| &\le \left|E(H_Q) \right| + \sum \limits_{(s, t) \in P \setminus Q} \left|\pi(s, t) \setminus E(H_Q) \right|\\
&\le f(n, p) + p \cdot \left( \frac{f(n, p)}{p} \right)\\
&= 2f(n, p).
\end{align*}
\end{proof}

This theorem implies that, if one uses the precomputed paths from Algorithm \ref{alg:knownpnonadaptive} to respond to online queries, then the online upper bound of $f(n, p)$ will still apply (up to a factor of $2$).
The main weakness of Algorithm \ref{alg:knownpnonadaptive} is that it requires advance knowledge of the parameter $p$, in order to compute the threshold $f(n, p) / p$ at which we exit the initial while loop.
It is tempting, but incorrect, to think this can be generally avoided by setting \emph{all} paths as in the main while loop.
Such a strategy would work for a path selection algorithm that happens to satisfy an axiom like \emph{monotonicity}, for which adding edges to $H$ can only decrease the number of new edges in a selected path $\pi(s, t \mid G, H)$ (it might also be fine to tolerate an approximate version of monotonicity).
However, we note that the path selection algorithms (\texttt{forwards-} and \texttt{backwards-growth}) used in our online upper bounds are \textbf{not} monotonic in this way (or even approximately monotonic).

That said, we next describe a wrapper for the algorithm that can avoid the need to know $p$ ahead of time:

\begin{algorithm}[t]
\textbf{Input:} $n$-node directed graph $G = (V, E)$\\

Let $f(n, p)$ be an extremal function for online reachability preservers, achieved by a deterministic path selection algorithm $\pi(s, t \mid G, H)$\;
Let $p^* := \arg \max_p f(n, p) \le O(n)$\;

\ForEach{$q \in \{p^*, 2p^*, 4p^*, 8p^*, \dots\}$}{
    Run Algorithm \ref{alg:knownpnonadaptive} with number-of-paths parameter $q$\;
    Denote selected paths by $\pi_q(s, t)$\;
}

\caption{\label{alg:indexnarpprep} Preprocessing for Index-Sensitive Non-Adaptive Reachability Preservers}
\end{algorithm}

\begin{algorithm} [t]
\textbf{Input:} demand pair $(s, t)$, index $i$\;

\tcp{run Algorithm \ref{alg:indexnarpprep} as preprocessing}

Let $q$ be the least value in $\{p^*, 2p^*, 4p^*, 8p^*, \dots\}$ with $q \ge i$\;
\textbf{Return} $\pi_q(s, t)$\;

\caption{\label{alg:indexnarpselect} Path Selection for Index-Sensitive Non-Adaptive Reachability Preservers}
\end{algorithm}

\FloatBarrier

\begin{theorem}
Suppose that the extremal function $f(n, p)$ depends polynomially on its second parameter $p$ in the regime where $p \ge p^*$.\footnote{This phrase ``depends polynomially'' is intuitive but a bit informal.  It is tedious to formalize it, but what we really mean is that this theorem holds for any function $f$ for which the latter two inequalities in the chain hold.  This includes all upper bounds shown in this paper.}
Then the online algorithm that runs Algorithm \ref{alg:indexnarpprep} as a preprocessing routine upon receiving $G$, and which then uses Algorithm \ref{alg:indexnarpselect} to select the path added to the preserver for each $i^{th}$ demand pair $(s, t)$, will construct a reachability preserver $H$ of size $|E(H)| \le O(f(n, p))$.
\end{theorem}
\begin{proof}
For each possible choice of $q$, we will add at most $q$ paths selected by $\pi_q$ to the preserver.
By Theorem \ref{thm:knownpnonadaptive}, these paths will have at most $2f(n, q)$ edges in their union.
Additionally, letting $q^*$ be the largest choice of $q$ for which we add any corresponding paths, note that we have $q^* \ge p \ge q^*/2$.
So we can bound the total number of edges in the preserver as:

\begin{align*}
|E(H)| &\le \sum \limits_{q \in \{p^*, 2p^*, 4p^*, 8p^*, \dots, q^*\}} 2f\left(n, q\right)\\
&\le O\left( f\left(n, q^* \right) \right)\\
&\le O\left( f\left(n, p\right) \right).
\end{align*}
Here the second inequality holds because $f$ depends polynomially on its second parameter, and so this sum is asymptotically dominated by its largest term.
The third inequality holds because we have $p \ge q^* / 2$, and (again since $f$ depends polynomially on its second parameter) this means the values of $f(n, q^*), f(n, p)$ differ by at most a constant factor.
\end{proof}
\section{Online Unweighted Directed Steiner Forest Algorithms \label{sec:onlineapp}}

We will next apply our extremal bounds for online reachability preservers to the problem of Online UDSF.
As a reminder, in this problem we receive an $n$-node directed graph $G = (V, E)$ on input, and then in each round we receive a new demand pair $(s, t)$ that is reachable in $G$.
We must irrevocably add edges to a reachability preserver $H$ to ensure that $(s, t)$ is reachable in $H$ before the next demand pair is received.
We do not know the number of demand pairs $p$ ahead of time.
We will denote by $OPT$ the size of the smallest possible (offline) reachability preserver for $G, P$, where $P$ is the set of all demand pairs received. 

\subsection{Recap of the Grigorescu-Lin-Quanrud Bound}

Our new bound will use the structure and several technical ingredients from the previous state-of-the-art online algorithm by Grigorescu, Lin, and Quanrud \cite{GLQ21}.
They proved:

\begin{theorem} [\cite{GLQ21}] \label{thm:oldnonlinest}
For online UDSF, there is a randomized polynomial time algorithm with competitive ratio $O(n^{2/3 + \eps})$.
\end{theorem}

Their algorithm carries two parameters, $T$ and $\tau$, which will be set at the end by a balance.
In the following, we will say that a demand pair $(s, t)$ is \emph{nontrivial} if it is not already reachable when it arrives, and thus it requires us to add at least one new edge to the preserver.
We let $p$ be the total number of nontrivial demand pairs.

\begin{itemize}
\item For the first $T$ nontrivial demand pairs that arrive, we use the following result by Chakrabarty et al.:

\begin{theorem} [\cite{ECKP15}] \label{thm:firstpairhandle}
There is a randomized polynomial-time online algorithm that constructs a preserver of the first $T$ demand pairs of size at most $OPT \cdot O(T^{1/2 - \eps})$.
\end{theorem}

\item After the first $T$ nontrivial demand pairs, the remaining demand pairs are further classified using a strategy from previous work on offline DSF \cite{CDKL17, BBMRY13}.
Say that a node pair $(s, t)$ is \emph{$\tau$-thin} if the number of vertices that lie along $s \leadsto t$ paths is at most $\tau$, or \emph{$\tau$-thick} otherwise.

\begin{itemize}
\item In order to handle the thick demand pairs, just after the $T^{th}$ demand pair is processed, we randomly sample a set of $|S| = Cn \log n/\tau$ nodes, where $C$ is a sufficiently large constant.
Let us say that a node pair $(s, t)$ is \emph{hit} by $S$ if there exists a node $v \in S$ that lies along an $s \leadsto t$ path.
By standard Chernoff bounds (omitted), with high probability, every thick pair $(s, t)$ is hit by $S$; in the following we will assume that this high-probability event occurs.
We then add an in- and out-tree from each sampled node in $S$, and so together these trees will contain an $s \leadsto t$ path.
This costs $\Oish(n^2 / \tau)$ edges in total.

\item When each demand pair $(s, t)$ arrives, we first check whether or not it is hit by $S$.
If so, the pair has been satisfied already by our trees and we can do nothing.
If not, then $(s, t)$ must be $\tau$-thin.
In this case, the algorithm appeals to an LP-rounding algorithm from \cite{BBMRY13}:
\begin{theorem} [\cite{BBMRY13}] \label{thm:thinhandle}
There is a randomized polynomial-time online algorithm that constructs a preserver of all $\tau$-thin pairs of size at most $OPT \cdot \Oish(\tau)$.
\end{theorem}
\end{itemize}
\end{itemize}

This completes the construction.
The last technical ingredient required is the following existential lower bound on OPT:
\begin{lemma} [\cite{GLQ21}] \label{lem:optlb}
$OPT \ge \Omega(p^{1/2})$.
\end{lemma}
\begin{proof}
Any set of $p$ demand pairs must use at least $p^{1/2}$ total terminal nodes (start or end).
In order to preserve connectivity, all terminal nodes must have in- or out-degree at least $1$.
It follows that any correct solution must have $\Omega(p^{1/2})$ edges.
\end{proof}

Now, assuming that $p \ge T$ we can bound the total competitive ratio as
\begin{align*}
&\frac{OPT \cdot O(T^{1/2 + \eps}) + \Oish\left(\frac{n^2}{\tau}\right) + OPT \cdot \Oish(t)}{OPT}\\
=& \ O(T^{1/2 + \eps}) + \Oish(\tau) + \Oish\left(\frac{n^2}{OPT \cdot \tau}\right)\\
\le& \ O(T^{1/2 + \eps}) + \Oish(\tau) + \Oish\left(\frac{n^2}{p^{1/2} \cdot \tau}\right)\\
\le& \ O(T^{1/2 + \eps}) + \Oish(\tau) + \Oish\left(\frac{n^{2}}{T^{1/2} \cdot \tau}\right).
\end{align*}
With a parameter balance, one can compute that the optimal setting is (essentially) $T = n^{4/3}$ and $\tau = n^{2/3}$, yielding the claimed competitive ratio of $O(n^{2/3 + \eps})$.
In the case where $p \le T$, the bound on competitive ratio is simply $O(T^{1/2 + \eps})$, which leads to the same bound.

\subsection{Our Adaptation}

We improve the competitive ratio from \cite{GLQ21}:
\begin{theorem} \label{thm:newnonlinest}
For online DSF with uniform costs, there is a randomized polynomial time algorithm with competitive ratio $O(n^{3/5 + \eps})$.
\end{theorem}

We mostly follow the strategy from \cite{GLQ21} outlined previously, but our main change is to the handling of thick pairs.
After the first $T$ demand pairs have been processed, we again sample a set $S$ of $|S| = \Oish(n/\tau)$ nodes, and we note that with high probability this sample hits all $\tau$-thick pairs.
However, unlike before, we do not add in- or out-trees from $S$.
Instead, when each new demand pair $(s, t)$ arrives, our strategy is to check whether or not it is hit by $S$.
If so, then we use the two cases of Theorem \ref{thm:onlinesource} to add paths for the pairs $(s, v)$ and $(v, t)$, at total cost $O((np|S|)^{1/2} + n)$ over all demand pairs.
If $(s, t)$ is not hit by $S$, then the new demand pair $(s, t)$ must be $\tau$-thin, and we handle it using Theorem \ref{thm:thinhandle} like before.

It is intuitive at this point that handling thick pairs with an improved bound would lead to an improved competitive ratio.
But unfortunately, it is not so simple: plugging in the improved bound to the previous competitive ratio will not lead to an improvement.
The reason is that the previous parameter settings of $T = n^{4/3}, \tau = n^{2/3}$ correspond to the setting where we have $|S| = \Oish(n^{1/3})$ nodes in our sample and $p=T=n^{4/3}$ demand pairs in total, and in this setting our new bounds roughly tie (up to $\log n$ factors) those obtained from using in- and out-trees.
To get around this, we will need a more nuanced version of Lemma \ref{lem:optlb}:

\begin{lemma} \label{lem:goodoptlb}
Let $p'$ be the number of nontrivial demand pairs that are hit by $S$.
Then $OPT \ge \Omega\left(p'/|S|\right)$. 
\end{lemma}
\begin{proof}
Similar to Lemma \ref{lem:optlb}, it suffices to argue that the demand pairs use $\Omega\left(p'/|S|\right)$ distinct terminal nodes (start or end), since each terminal node must have (in or out) degree at least $1$.

For every demand pair $(s, t)$ that is hit by $S$, there is a node $v \in S$ for which we add $s \leadsto v$ and $v \leadsto t$ paths to the preserver.
This must either be the first time we add an $s \leadsto v$ path, or the first time we add a $v \leadsto t$ path, since otherwise an $s \leadsto t$ path in the preserver will already exist and the demand pair will be trivial.
Thus the number of start or end terminals that have been paired with $v$ increases by $1$.
Overall, since there are $|S|$ possible nodes that could hit our demand pairs, we must have at least $p'/|S|$ terminal nodes in total.
\end{proof}

We are now ready to calculate competitive ratio.
Assuming that $p \ge T$, this is

\begin{align*}
& \ O\left(T^{1/2 + \eps}\right) + \Oish(\tau) + O\left(\frac{|S|^{1/2} n^{1/2} p'^{1/2}}{OPT}\right) + O\left(\frac{n}{OPT}\right)\\
\le& \ O\left(T^{1/2 + \eps}\right) + \Oish(\tau) + O\left(\frac{|S|^{1/2} n^{1/2} p'^{1/2}}{\max\{p'/|S|, p^{1/2}\}}\right) + O\left(\frac{n}{p^{1/2}}\right) \tag*{Lemmas \ref{lem:optlb}, \ref{lem:goodoptlb}}\\
=& \ O\left(T^{1/2 + \eps}\right) + \Oish(\tau) + O\left(\min\left\{\frac{|S|^{3/2} n^{1/2}}{p'^{1/2}} , \frac{|S|^{1/2} n^{1/2} p'^{1/2}}{p^{1/2}} \right\}\right) + O\left(\frac{n}{p^{1/2}}\right)\\
=& \ O\left(T^{1/2 + \eps}\right) + \Oish(\tau) + \Oish\left(\min\left\{\frac{n^2}{p'^{1/2} \tau^{3/2}} , \frac{n p'^{1/2}}{p^{1/2} \tau^{1/2}} \right\}\right) + O\left(\frac{n}{p^{1/2}}\right) \tag*{$|S| = \Oish(n / \tau)$}\\
\le& \ O\left(T^{1/2 + \eps}\right) + \Oish(\tau) + \Oish\left(\min\left\{\frac{n^2}{p'^{1/2} \tau^{3/2}} , \frac{n p'^{1/2}}{T^{1/2} \tau^{1/2}} \right\}\right) + O\left(\frac{n}{T^{1/2}}\right). \tag*{$p \ge T$}
\end{align*}
This bound will be maximized when $p'$ is such that the two terms in the $\min$ balance, which occurs when
$$p' = n \cdot \frac{T^{1/2}}{\tau}.$$
Under this setting, we can simplify
\begin{align*}
\le& \ O\left(T^{1/2 + \eps}\right) + \Oish(\tau) + \Oish\left( \frac{n^{3/2}}{T^{1/4} \tau}\right) + O\left(\frac{n}{T^{1/2}}\right).
\end{align*}
Finally, we are ready to choose $\tau, T$ to balance these terms.
By setting
$$\tau := n^{3/5}, T := n^{6/5},$$
the above expression becomes
\begin{align*}
& \ O\left(n^{3/5 + \eps}\right) + \Oish(n^{3/5}) + \Oish\left( \frac{n^{3/2}}{n^{3/10} \cdot n^{3/5}}\right) + O\left(\frac{n}{n^{3/5}}\right)\\
=& \ O\left(n^{3/5 + \eps}\right).
\end{align*}
Finally, as before, in the case where $p \le T$ the competitive ratio is $O(T^{1/2 + \eps})$, giving the same bound.

\bibliographystyle{plain}
\bibliography{refs}

\begin{thebibliography}{10}

\bibitem{AB18}
Amir Abboud and Greg Bodwin.
\newblock Reachability preservers: New extremal bounds and approximation
  algorithms.
\newblock In {\em Proceedings of the 29th Annual ACM-SIAM Symposium on Discrete
  Algorithms (SODA)}, pages 1865--1883. Society for Industrial and Applied
  Mathematics, 2018.

\bibitem{Alon02}
Noga Alon.
\newblock Testing subgraphs in large graphs.
\newblock {\em Random Structures \& Algorithms}, 21(3-4):359--370, 2002.

\bibitem{AAA03}
Noga Alon, Baruch Awerbuch, and Yossi Azar.
\newblock The online set cover problem.
\newblock In {\em Proceedings of the thirty-fifth annual ACM symposium on
  Theory of computing}, pages 100--105, 2003.

\bibitem{AAABN06}
Noga Alon, Baruch Awerbuch, Yossi Azar, Niv Buchbinder, and Joseph Naor.
\newblock A general approach to online network optimization problems.
\newblock {\em ACM Transactions on Algorithms (TALG)}, 2(4):640--660, 2006.

\bibitem{BCR16}
Surender Baswana, Keerti Choudhary, and Liam Roditty.
\newblock Fault tolerant subgraph for single source reachability: generic and
  optimal.
\newblock In {\em Proceedings of the forty-eighth annual ACM symposium on
  Theory of Computing}, pages 509--518. ACM, 2016.

\bibitem{BBMRY13}
Piotr Berman, Arnab Bhattacharyya, Konstantin Makarychev, Sofya Raskhodnikova,
  and Grigory Yaroslavtsev.
\newblock Approximation algorithms for spanner problems and directed steiner
  forest.
\newblock {\em Information and Computation}, 222:93--107, 2013.

\bibitem{Bodwin21}
Greg Bodwin.
\newblock New results on linear size distance preservers.
\newblock {\em SIAM Journal on Computing}, 50(2):662--673, 2021.

\bibitem{BHT23}
Greg Bodwin, Gary Hoppenworth, and Ohad Trabelsi.
\newblock Bridge girth: A unifying notion in network design.
\newblock In {\em 2023 IEEE 64th Annual Symposium on Foundations of Computer
  Science (FOCS)}, pages 600--648. IEEE, 2023.

\bibitem{BN06}
Niv Buchbinder and Joseph Naor.
\newblock Improved bounds for online routing and packing via a primal-dual
  approach.
\newblock In {\em 2006 47th Annual IEEE Symposium on Foundations of Computer
  Science (FOCS'06)}, pages 293--304. IEEE, 2006.

\bibitem{BN09}
Niv Buchbinder and Joseph Naor.
\newblock Online primal-dual algorithms for covering and packing.
\newblock {\em Mathematics of Operations Research}, 34(2):270--286, 2009.

\bibitem{BN+09}
Niv Buchbinder, Joseph~Seffi Naor, et~al.
\newblock The design of competitive online algorithms via a primal--dual
  approach.
\newblock {\em Foundations and Trends{\textregistered} in Theoretical Computer
  Science}, 3(2--3):93--263, 2009.

\bibitem{CCC22}
Diptarka Chakraborty, Kushagra Chatterjee, and Keerti Choudhary.
\newblock {Pairwise Reachability Oracles and Preservers Under Failures}.
\newblock In Miko{\l}aj Boja\'{n}czyk, Emanuela Merelli, and David~P. Woodruff,
  editors, {\em 49th International Colloquium on Automata, Languages, and
  Programming (ICALP 2022)}, volume 229 of {\em Leibniz International
  Proceedings in Informatics (LIPIcs)}, pages 35:1--35:16, Dagstuhl, Germany,
  2022. Schloss Dagstuhl -- Leibniz-Zentrum f{\"u}r Informatik.

\bibitem{CC20}
Diptarka Chakraborty and Keerti Choudhary.
\newblock {New Extremal Bounds for Reachability and Strong-Connectivity
  Preservers Under Failures}.
\newblock In Artur Czumaj, Anuj Dawar, and Emanuela Merelli, editors, {\em 47th
  International Colloquium on Automata, Languages, and Programming (ICALP
  2020)}, volume 168 of {\em Leibniz International Proceedings in Informatics
  (LIPIcs)}, pages 25:1--25:20, Dagstuhl, Germany, 2020. Schloss
  Dagstuhl--Leibniz-Zentrum f{\"u}r Informatik.

\bibitem{CCCDGGL99}
Moses Charikar, Chandra Chekuri, To-Yat Cheung, Zuo Dai, Ashish Goel, Sudipto
  Guha, and Ming Li.
\newblock Approximation algorithms for directed steiner problems.
\newblock {\em Journal of Algorithms}, 33(1):73--91, 1999.

\bibitem{CEGS11}
Chandra Chekuri, Guy Even, Anupam Gupta, and Danny Segev.
\newblock Set connectivity problems in undirected graphs and the directed
  steiner network problem.
\newblock {\em ACM Transactions on Algorithms (TALG)}, 7(2):1--17, 2011.

\bibitem{CDKL17}
Eden Chlamt{\'a}{\v{c}}, Michael Dinitz, Guy Kortsarz, and Bundit Laekhanukit.
\newblock Approximating spanners and directed steiner forest: Upper and lower
  bounds.
\newblock In {\em Proceedings of the Twenty-Eighth Annual ACM-SIAM Symposium on
  Discrete Algorithms}, pages 534--553. SIAM, 2017.

\bibitem{Choudhary16}
Keerti Choudhary.
\newblock An optimal dual fault tolerant reachability oracle.
\newblock In {\em LIPIcs-Leibniz International Proceedings in Informatics},
  volume~55. Schloss Dagstuhl-Leibniz-Zentrum fuer Informatik, 2016.

\bibitem{chuzhoy2009polynomial}
Julia Chuzhoy and Sanjeev Khanna.
\newblock Polynomial flow-cut gaps and hardness of directed cut problems.
\newblock {\em Journal of the ACM (JACM)}, 56(2):1--28, 2009.

\bibitem{CE06}
Don Coppersmith and Michael Elkin.
\newblock Sparse sourcewise and pairwise distance preservers.
\newblock {\em SIAM Journal on Discrete Mathematics}, 20(2):463--501, 2006.

\bibitem{DK99}
Yevgeniy Dodis and Sanjeev Khanna.
\newblock Design networks with bounded pairwise distance.
\newblock In {\em Proceedings of the thirty-first annual ACM symposium on
  Theory of computing}, pages 750--759, 1999.

\bibitem{ECKP15}
Alina Ene, Deeparnab Chakrabarty, Ravishankar Krishnaswamy, and Debmalya
  Panigrahi.
\newblock Online buy-at-bulk network design.
\newblock In {\em 2015 IEEE 56th Annual Symposium on Foundations of Computer
  Science}, pages 545--562. IEEE, 2015.

\bibitem{FKN12}
Moran Feldman, Guy Kortsarz, and Zeev Nutov.
\newblock Improved approximation algorithms for directed steiner forest.
\newblock {\em Journal of Computer and System Sciences}, 78(1):279--292, 2012.

\bibitem{GW95}
Michel~X Goemans and David~P Williamson.
\newblock A general approximation technique for constrained forest problems.
\newblock {\em SIAM Journal on Computing}, 24(2):296--317, 1995.

\bibitem{GLQ21}
Elena Grigorescu, Young-San Lin, and Kent Quanrud.
\newblock {Online Directed Spanners and Steiner Forests}.
\newblock In Mary Wootters and Laura Sanit\`{a}, editors, {\em Approximation,
  Randomization, and Combinatorial Optimization. Algorithms and Techniques
  (APPROX/RANDOM 2021)}, volume 207 of {\em Leibniz International Proceedings
  in Informatics (LIPIcs)}, pages 5:1--5:25, Dagstuhl, Germany, 2021. Schloss
  Dagstuhl -- Leibniz-Zentrum f{\"u}r Informatik.

\bibitem{HP18}
Shang-En Huang and Seth Pettie.
\newblock {Lower Bounds on Sparse Spanners, Emulators, and Diameter-reducing
  shortcuts}.
\newblock In David Eppstein, editor, {\em 16th Scandinavian Symposium and
  Workshops on Algorithm Theory (SWAT 2018)}, volume 101 of {\em Leibniz
  International Proceedings in Informatics (LIPIcs)}, pages 26:1--26:12,
  Dagstuhl, Germany, 2018. Schloss Dagstuhl--Leibniz-Zentrum fuer Informatik.

\bibitem{Winter87}
Pawel Winter.
\newblock Steiner problem in networks: A survey.
\newblock {\em Networks}, 17(2):129--167, 1987.

\end{thebibliography}

\appendix
\section{Online Pairwise Reachability Preservers \label{app:onlinerp}}

We will next prove Theorem \ref{thm:intropairwise}.
Recall from Lemma \ref{lem:zprops} that, if we construct our online reachability preserver using \texttt{forwards-} or \texttt{backwards-growth}, then it suffices to bound the size of the associated path system $Z$ as defined in Section \ref{sec:pathgrowth}.
In particular, under \texttt{backwards-growth}, $Z$ will have the following property:
\begin{definition} [Half-Bridge-Freeness]
    A ordered path system is said to be half-$k$-bridge-free if there are no bridges of size at most $k$ in which the last arc comes before the river.
\end{definition}

The focus of our proof will shift to bounding the maximum possible size of \emph{any} half-bridge-free system.
That is, let $H(n, p, k)$ denotes the maximum size of a half-$k$-bridge-free system with $n$ nodes and $p$ paths, and then from Lemma \ref{lem:zprops} we have
$$\|Z\|\le H(n, p, \infty) \le H(n, p, 4).$$
So we may focus on providing an upper bound for $H(n, p, 4)$.

\subsection{Proof Overview \label{app:overview}}

The proof gets quite technical in places, so let us start with a higher-level overview of the proof strategy. 
Let $Z = (V, \Pi)$ be a system with $n$-nodes, $p$-paths, and no half $2, 3,$ or $4$ bridges.
By the standard \emph{cleaning lemma} (Lemma \ref{lem:cleaning}), we may assume that all nodes have degree $\Theta(d)$, and all paths have length $\Theta(\ell)$.
For simplicity, we will assume in this overview that all nodes have degree \emph{exactly} $d$ and all paths have length \emph{exactly} $\ell$, which will not materially affect the argument.

\paragraph{Recap of Offline Proof from \cite{BHT23}.}

The previous-best offline upper bound from \cite{BHT23} focuses on a collection of sets $\left\{ R(\pi_1, \pi_3) \ \mid \ \pi_1, \pi_3 \in \Pi\right\}$ which are each the subsets of paths from $\Pi$ that intersect $\pi_1$ and then later intersect $\pi_3$:

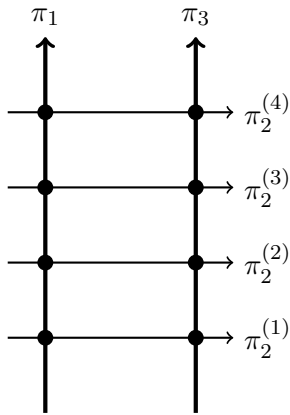
\begin{center}
\begin{tikzpicture}
    \draw[ultra thick, ->] (0, 0) -- (0, 5) node[above] {$\pi_1$};
    \draw[ultra thick, ->] (2, 0) -- (2, 5) node[above] {$\pi_3$};

    \foreach \y in {1, 2, 3, 4} {
        \draw[thick, ->] (-0.5, \y) -- (2.5, \y) node[right] {$\pi_2^{(\y)}$};
        \fill (0, \y) circle (3pt);
        \fill (2, \y) circle (3pt);
    }

    \node at (8, 2.5) {$R(\pi_1, \pi_3) = \left\{\pi_2^{(1)}, \pi_2^{(2)}, \pi_2^{(3)}, \pi_2^{(4)}\right\}$};
\end{tikzpicture}
\end{center}

Naively, these sets can have maximum size $|R(\pi_1, \pi_3)| \le \ell^2$, since there are $\ell^2$ ways to choose a node from $\pi_1$ and then $\pi_3$, and no two paths can use the same pair of intersection points (or else they form a forbidden $2$-bridge).
Some straightforward algebra from there leads to an initial but very suboptimal bound of

$$\|Z\| \le O\left( \min\left\{np^{1/2}, n^{1/2}p\right\} + n + p\right).$$

The next improvement comes by observing that we can actually only pack $\ell$ paths between $\pi_1$ and $\pi_3$, rather than $\ell^2$.
This is essentially because the paths in $R(\pi_1, \pi_3)$ cannot cross each other, as in the following picture, or else they will form a forbidden bridge:
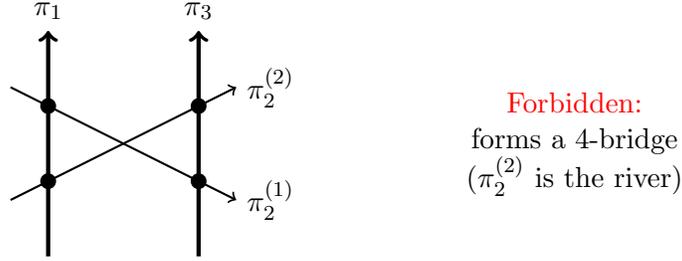
\begin{center}
\begin{tikzpicture}
    \draw[ultra thick, ->] (0, 0) -- (0, 3) node[above] {$\pi_1$};
    \draw[ultra thick, ->] (2, 0) -- (2, 3) node[above] {$\pi_3$};

    \draw[thick, ->] (-0.5, 0.75) -- (2.5, 2.25) node[right] {$\pi_2^{(2)}$};
    \draw[thick, ->] (-0.5, 2.25) -- (2.5, 0.75) node[right] {$\pi_2^{(1)}$};

    \fill (0, 1) circle (3pt);
    \fill (2, 1) circle (3pt);
    \fill (0, 2) circle (3pt);
    \fill (2, 2) circle (3pt);

    \node [align=center] at (7, 1.5) {\color{red} Forbidden:\\forms a $4$-bridge\\($\pi_2^{(2)}$ is the river)};
\end{tikzpicture}
\end{center}

Redoing the algebra with this improved bound $|R(\pi_1, \pi_3)| \le \ell$ leads to a better bound of
$$ \|Z\| \le O\left( (np)^{2/3} + n + p\right). $$

It is indeed possible for \emph{some} sets to have size $|R(\pi_1, \pi_3)| = \ell$, but the next round of improvements works by arguing that the \emph{typical} such set must be slightly smaller.
For intuition: suppose towards contradiction that every set has size exactly $|R(\pi_1, \pi_3)| = \ell$.
This can occur only if every path $\pi_2 \in R(\pi_1, \pi_3)$ is perfectly \emph{aligned} with $\pi_1$ and $\pi_3$: that is, if $\pi_2$ intersects $\pi_1$ on its $i^{th}$ node, then it also intersects $\pi_3$ on its $i^{th}$ node.
But then - under this assumption of alignment - we can argue that our paths must be unusually concentrated.
Consider an arbitrary \emph{base path} $\pi_b$, and consider the induced subsystem of nodes that come one step after $\pi_b$, along paths that intersect $\pi_b$:

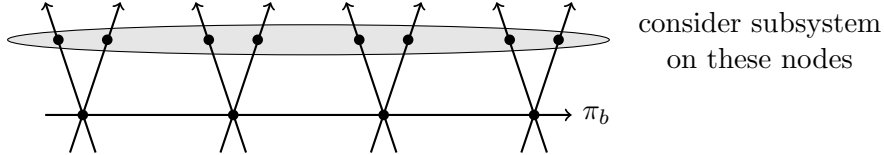
\begin{center}
\begin{tikzpicture}

    \draw [black, fill=gray!20] (3, 1) ellipse (4cm and 0.2cm);
    \node [align=center] at (9, 1) {consider subsystem\\on these nodes};
    
    \draw[->, thick] (-0.5, 0) -- (6.5, 0) node[right] {$\pi_b$};

    \foreach \x in {0, 2, 4, 6} {
        \draw[->, thick] ({\x+0.17}, -0.5) -- ({\x - 0.5}, 1.5);
        \fill (\x, 0) circle (2pt);  
        \fill ({\x - 0.325}, 1) circle (2pt); 

        \draw[->, thick] ({\x-0.17}, -0.5) -- ({\x + 0.5}, 1.5);
        \fill ({\x + 0.325}, 1) circle (2pt); 
    }

\end{tikzpicture}
\end{center}

This subsystem will have $\ell d$ nodes, node degrees $d$, and (due to alignment) the paths in this subsystem will have length $\ell$.
However, applying the previously-shown bounds to this subsystem reveals that these particular parameters imply a contradiction: the typical length of a path in this induced subsystem must actually be $\ll \ell$, giving contradiction.

Formalizing this intuition gets rather technical.
The key moving parts that were not covered in the above sketch are:
\begin{itemize}
\item We need to select the base path $\pi_b$ \emph{randomly} rather than arbitrarily, and then measure the \emph{expected} value of the various statistics of the system.
\item We focus on a subsystem formed by $h$ nodes along paths intersecting $\pi_b$, rather than just $1$ node.
(Here $h$ is a new parameter, set by a parameter balance at the end of the proof.)
\item The proof involves toggling between bounding the sum of path lengths (which is $\|Z\|$), and bounding the sum of \emph{squared} path lengths.
This is necessary because (1) the sum of sizes of $R(\pi_1, \pi_3)$ sets naturally scales with the sum of squared path lengths, rather than the sum of path lengths, and (2) when we focus on our induced subsystem, it may break the cleaning lemma: it is not still guaranteed that all paths have the same length \emph{when restricted to that subsystem}.
\end{itemize}

\paragraph{A New Optimization.}

One of the ways the current proof differs from \cite{BHT23} is in an improved handling of the toggling between sum-of-path-lengths and sum-of-squared-path-lengths.
Roughly, instead of switching back and forth between these as the proof proceeds, the proof focuses almost entirely on sum-of-squared-path-lengths.
This has several advantages in efficiency; perhaps most notably, it lets us \emph{recursively} apply our bound on the size of $Z$ to control the size of the induced subsystem, rather than only applying the bound $\|Z\|\le O((np)^{2/3} + n + p)$ once.
This is what leads to our improvements in the offline setting.
This is also mildly helpful in the online setting, but there are other necessary changes that lead to losses that more than eclipse this improvement.

\paragraph{Changes in the Online Setting.}

The first change that is required in the online setting is in the definition of the sets $R(\pi_1, \pi_3)$.
Instead of including all paths $\pi_2$ in these sets that intersect $\pi_1$ followed by $\pi_3$, we only include such paths $\pi_2$ that come after both $\pi_1$ and $\pi_3$ in the ordering.
This restriction only affects the relevant counting by constant factors, and it also suffices to achieve the crucial total-ordering property of the $R(\pi_1, \pi_3)$ sets:

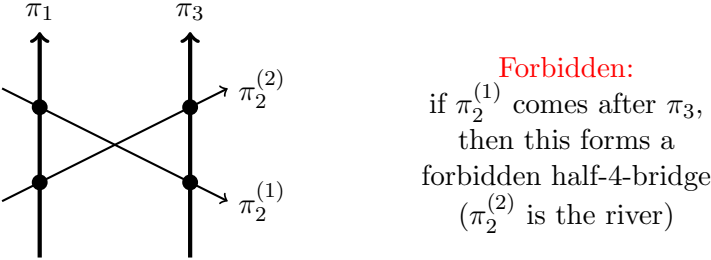
\begin{center}
\begin{tikzpicture}
    \draw[ultra thick, ->] (0, 0) -- (0, 3) node[above] {$\pi_1$};
    \draw[ultra thick, ->] (2, 0) -- (2, 3) node[above] {$\pi_3$};

    \draw[thick, ->] (-0.5, 0.75) -- (2.5, 2.25) node[right] {$\pi_2^{(2)}$};
    \draw[thick, ->] (-0.5, 2.25) -- (2.5, 0.75) node[right] {$\pi_2^{(1)}$};

    \fill (0, 1) circle (3pt);
    \fill (2, 1) circle (3pt);
    \fill (0, 2) circle (3pt);
    \fill (2, 2) circle (3pt);

    \node [align=center] at (7, 1.5) {\color{red} Forbidden:\\if $\pi_2^{(1)}$ comes after $\pi_3$,\\then this forms a\\forbidden half-$4$-bridge\\($\pi_2^{(2)}$ is the river)};
\end{tikzpicture}
\end{center}

A more serious problem arises when considering the induced subsystem.
In the offline setting, the induced subsystem has $\ell d h$ total nodes.
This holds because there are $\ell d$ paths that branch from $\pi_b$, and these paths may not intersect each other, or else they form a $3$-bridge. 

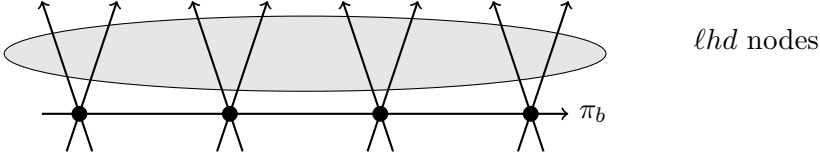
\begin{center}
\begin{tikzpicture}

    \draw [black, fill=gray!20] (3, 0.8) ellipse (4cm and 0.5cm);
    \node [align=center] at (9, 1) {$\ell h d$ nodes};
    
    \draw[->, thick] (-0.5, 0) -- (6.5, 0) node[right] {$\pi_b$};

    \foreach \x in {0, 2, 4, 6} {
        \draw[->, thick] ({\x+0.17}, -0.5) -- ({\x - 0.5}, 1.5);
        \fill (\x, 0) circle (3pt);  

        \draw[->, thick] ({\x-0.17}, -0.5) -- ({\x + 0.5}, 1.5);
    }

\end{tikzpicture}
\end{center}

With only half-bridges forbidden, however, it is possible for these branching paths to intersect each other.
This will form a $3$-bridge, but it will \emph{not} necessarily form a half-$3$-bridge.
This might lead to significantly fewer than $\ell d h$ nodes in the induced subsystem.

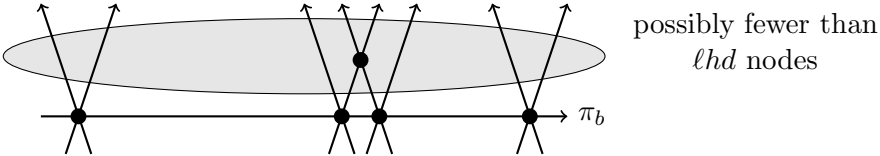
\begin{center}
\begin{tikzpicture}

    \draw [black, fill=gray!20] (3, 0.8) ellipse (4cm and 0.5cm);
    \node [align=center] at (9, 1) {possibly fewer than\\$\ell h d$ nodes};
    
    \draw[->, thick] (-0.5, 0) -- (6.5, 0) node[right] {$\pi_b$};

    \foreach \x in {0, 3.5, 4, 6} {
        \draw[->, thick] ({\x+0.17}, -0.5) -- ({\x - 0.5}, 1.5);
        \fill (\x, 0) circle (3pt);  

        \draw[->, thick] ({\x-0.17}, -0.5) -- ({\x + 0.5}, 1.5);
    }

    \fill (3.75, 0.75) circle (3pt);

\end{tikzpicture}
\end{center}

A lot of our new technical work is to control the amount of overlap that might occur, thus giving a nontrivial lower bound on the number of nodes in the induced subsystem.
Roughly, we do this by recursively applying our bounds on the size of a half-bridge-free system in yet another way.
Still though, the lower bound on number of nodes is still $\ll ld h$, and this is the fundamental reason why our bounds in the online setting are polynomially worse than the analogous bounds in the offline setting.

\subsection{Setup of Formal Proof and Initial Bound}

We start with the following standard lemma:
\begin{lemma}[Cleaning Lemma -- c.f.~\cite{BHT23}, Lemma 10]\label{lem:cleaning}
    There exists a half-$k$-bridge free system on $\le n$ vertices, $\le p$ path whose size is $\Theta(H(n, p, k))$ such that every vertex has degree between $d/4$ and $4d$ and every path has length between $\ell/4$ and $4\ell$ where $d, \ell$ are the average degree and average path length respectively.
\end{lemma}
\begin{proof}
Start with a half-$k$-bridge free system $S$ with $n$ nodes, $p$ paths, and size $\beta(n, p, k)$.
Fix $d$ as the average node degree and $\ell$ as the average path length (which do not change as we modify the system).
Then, perform of the following steps:
\begin{itemize}
\item While there exists a node of degree $< d/4$ or a path of degree $< \ell/4$, remove the node or path from the system.

\item While there exists a node $v$ of degree $> d$, split it into two nodes $\{v_1, v_2\}$. Replace each instance of $v$ in a path with either $v_1$ or $v_2$ in such a way that $\deg(v_1) \in \{deg(v_2), deg(v_2)+1\}$.

\item While there exists a path $\pi$ of length $> \ell$, split it into two node-disjoint subpaths $\pi_1, \pi_2$, with $|\pi_1| \in \{|\pi_2|, |\pi_2|+1\}$.
\end{itemize}
It is clear from the construction that all remaining nodes have degree in the range $[d/4, d]$, and that all remaining paths have length in the range $[\ell/4, \ell]$, and that the modified system is still half-$k$-bridge-free.
Additionally, by unioning over the nodes and paths, the overall size of the system decreases by $< nd/4 + p\ell/4 = \|S\|/2$ due to deletions, so the size is still $\Theta(\beta(n, p, k))$, and the lemma is satisfied.
\end{proof}

Let $\ell$ and $d$ denotes the average path length and average node degree respectively, then we have $\norm{Z} = nd = p\ell$. For the rest of this section, we assume $\ell$ and $d$ are at least a sufficiently large constant, otherwise we immediately have $\norm{Z}\le O(n+p)$.
By the cleaning lemma, we may assume that every vertex has degree between $d/4$ and $4d$ and every path has length between $\ell/4$ and $4\ell$.

We write $\pi_1 < \pi_2$ if the path $\pi_1$ comes before $\pi_2$ in the ordering. For nodes $a, b\in \pi_1$, write $a <_{\pi_1} b$ if $a$ comes before $b$ in the path $\pi_1.$

\begin{definition} [$R$ Sets]
    Let
    \[R= \{(\pi_1,\pi_2,\pi_3): \pi_1 \cap\pi_2 <_{\pi_2} \pi_2 \cap\pi_3, \pi_3 < \pi_2\}.\] 
    For any two paths $\pi_1, \pi_3$, we define
    $$R(\pi_1, \pi_3):= \{\pi_2: (\pi_1, \pi_2, \pi_3)\in R\}.$$
\end{definition}
\begin{lemma} [c.f. \cite{BHT23}, Section 3.2.4] \label{lem:rsorder}
    For any pair of paths $\pi_1, \pi_3$, there is a total ordering of the elements of $R(\pi_1, \pi_3)$, denoted $ <_R$ such that if $\pi_a<_{R} \pi_b$ then $\pi_a \cap \pi_1 <_{\pi_1}\pi_b\cap \pi_1$ and $\pi_a \cap \pi_3 \le_{\pi_3}\pi_b\cap \pi_3$. As a corollary, we have $|R(\pi_1, \pi_3)|\le O(\ell)$.
\end{lemma}
    \begin{proof}
        We first note that for $\pi_a, \pi_b\in  R(\pi_1, \pi_3)$ we have if $\pi_a\ne \pi_b$ then $\pi_a\cap \pi_1 \ne \pi_b\cap \pi_1$. Suppose not, and without loss of generality assume $\pi_a\cap \pi_3\le_{\pi_3} \pi_b\cap \pi_3$. If $\pi_a\cap \pi_3 =  \pi_b\cap \pi_3$ then $\pi_a, \pi_b$ forms a $2$-bridge, which is a contradiction. Otherwise, $\pi_a\cap \pi_3 <_{\pi_3} \pi_b \cap \pi_3$, so $\pi_a , \pi_3, \pi_b$ forms a $3$ bridge with $\pi_3$ being the last arc and $\pi_b$ is the river and  $\pi_3 < \pi_b$, which is also a contradiction.

        Thus we can order the elements of $R(\pi_1, \pi_3)$ by $\pi_a <_R \pi_b$ if $\pi_1\cap\pi_a <_{\pi_1}\pi_1\cap\pi_b$. It suffices to show that if $\pi_1\cap\pi_a <_{\pi_1}\pi_1\cap\pi_b$ then $\pi_a\cap\pi_3 \le_{\pi_3} \pi_b \cap \pi_3$. Suppose, for contradiction, that there is $\pi_a, \pi_b$ such that $\pi_1\cap\pi_a <_{\pi_1}\pi_1\cap\pi_b$ but $\pi_a\cap\pi_3 >_{\pi_3}\pi_b \cap \pi_3$. Then we have $\pi_1, \pi_b, \pi_3, \pi_a$ forms a $4$-bridge with $\pi_3$ being the last arc and $\pi_a$ is the river, and $\pi_3 < \pi_a$, contradiction.
    \end{proof}

\begin{lemma}\label{maxaverage}
    Let $x_1, \dots, x_n$ be numbers  that are at most $m$ and the average is at least $a$. Then for any $t  < a$ the number of numbers that is at least $t$ is at least $\frac{a-t}{m-a}\cdot n$.
\end{lemma}
    \begin{proof}
        We apply Markov's inequality to the random variable $m - x$ where $x$ is sampled uniformly from $x_1, \dots, x_n$.
        This gives
        $$\Pr\left[m-x\ge m-a\right] \ge \frac{a-t}{m-a},$$
        so the number of numbers that is at least $t$ must be at least $\frac{a-t}{m-a}\cdot n$.
    \end{proof}

\begin{lemma} [Rephrasing of \cite{BHT23}, Lemma 23] \label{lem:rssize}
$|R|\ge \Omega\pa{p\ell^2 d^2}$
\end{lemma}    
    \begin{proof}
        For a vertex $u$ let \[f(u) := \sum_{\pi\ni u}|\{y\in \pi: u>_\pi y\}|,\]
        e.g. the number of vertices that strictly precede $u$ in some path. Note that \[\frac1n\sum_{u\in V}f(u) = \frac1n\sum_{\pi\in \Pi} \binom{|\pi|}{2}\ge \frac1np \binom{\ell/4}{2}\ge \frac1n\cdot\frac{p\ell^2}{33} = \frac{d\ell}{33}\]
        while for each $u$ we have $f(u)\le 4d\cdot 4\ell = 16d\ell$. Thus by Lemma \ref{maxaverage} there is at least $\frac n{1054} = \Omega(n)$ vertices $u$ such that $f(u)\ge \frac{d\ell}{66}$. Call such a vertex good.
        
        Now consider any good vertex $u$. Let $d_u$ be degree of $u$. For each $\pi\ni u$ define $f(u, \pi)$ to be the number of vertex that strictly precede $u$ in $\pi$. Then we have \[\frac1{d_u}\sum_{\pi\ni u}f(u, \pi) = \frac1{d_u}f(u) \ge \frac {d\ell}{66 d_u}\ge \frac{\ell}{132}\]
        while for each $\pi$ we have $f(u, \pi) \le 2\ell$. So by Lemma \ref{maxaverage} there is at least $\Omega(d_u) = \Omega(d)$ path $\pi\ni u$ such that $f(u, \pi)\ge \frac{\ell}{200}$. Call such a path good. Let $d_g$ be the number of good paths (through $u$.)

        Now consider any pair $\pi_2, \pi_3$ of good path. Note that $\pi_2 < \pi_3$ or $\pi_3 < \pi_2$ by assumption. Thus there are at least $\binom{d_g}2 = \Omega(d^2)$ pair of good path $\pi_3 < \pi_2$ for each good vertex $u$. 


        Now for a good pair $\pi_3<  \pi_2$, consider all the vertices preceding $u$ in $\pi_2$, there are at least $\frac{\ell}{36} = \Omega(\ell)$ such vertices by the definition of good path. Fix such a vertex $v$. Each such $v$ has $\Omega(d)$ path $\pi_1$ passing through them. As we goes through $\Omega(n)$ good vertex $u$, each with $\Omega(d^2)$ pair of good path $\pi_3 < \pi_2$ passing through $u$, and each with $\Omega(\ell)$ vertex $v$ such that $v <_{\pi_2} u$, each with $\Omega(d)$ path $\pi_1$ passing through $v$, we have at least $\Omega(nd^2 \ell d) = \Omega(p\ell^2 d^2)$ elements in the set $R$.
    \end{proof}
\begin{lemma}[Initial bound]\label{lem:initialbound}
    We have $H(n, p, 4)\le O(n + n^{2/3}p^{2/3} + p)$.
    \begin{proof}
        We have
        \begin{equation*}
            \Omega(p\ell^2d^2)\le |R|\le \sum_{\pi_1, \pi_3\in \Pi}|R(\pi_1, \pi_3)|\le O(p^2 \ell).
        \end{equation*}
        Rearranging gives $\norm{Z}\le n^{2/3}p^{2/3}$.
    \end{proof}
\end{lemma}

\subsection{Bootstrapping}

Our strategy from now on will be the following. Starting with the initial bound above, we will recursively improve the bound. For the recursive improvement step, we sample a path, called the \textit{base path}  uniformly at random from $\Pi$. For each of the $\Theta(\ell)$ vertices on the base path, consider the $\Theta(d)$ paths that pass through it. Then for each such path consider the $h$ vertices that immediately follow the vertex on the base path, where $h\le \ell$ is a parameter that will be chosen later (if there are few than $h$ vertices following the base path vertex, consider all of them.) Each such sequence of at most $O(h)$ vertices is called a branching path. Thus at most $O(\ell dh)$ vertices are considered, and consider the subsystem induced by these vertices, which we shall call the $h$-system induced  by the base path. Let $Q$ denotes the set of tuple $(\pi_1, \pi_3, u, v)$ such that $\pi_1, \pi_3$ are branching paths, $u\in \pi_1, v\in \pi_3$, $u\notin \pi_b$ ($\pi_b$ is the base path), and $u <_{\pi_2} v$ for some other path $\pi_2$, and $\pi_3 < \pi_2$. Then we will argue that the structural property of the graph force $|Q|$ to be high in expectation, that is, we prove a lower bound of $Q$ in term of $\ell, d, n, p, h$. Then, we shall use the current bound for $H(n, p, 4)$, applied appropriately on parts of the $h$-system, to prove an upper bound for $|Q|$. Comparing the upper and lower bound for $|Q|$ would lead to a better bound for $H(n, p, 4)$, which converges to the final bound claimed.

Let us start with the lower bound for $Q$.

\begin{lemma} [c.f. \cite{BHT23}, Lemma 24] \label{lem:lowerboundq}
    If $\ell\ge h\ge \frac{Cp}{\ell d^2}$ for some large enough constant $C$, we have
    \begin{equation*}
        \E[|Q|]\ge \Omega\pa{\frac{h}{\ell p}\sum_{\pi_1, \pi_3\in \Pi} \abs*{R(\pi_1, \pi_3)}^2}.
    \end{equation*}
    \begin{proof}
    For any vertices $a$ and path $\pi$ such that $a\in \pi$, let $\pi[a]$ denotes the number of vertices in $\pi$ that is weakly before $a$.
    Fix a pair $\pi_1, \pi_3$.  
    Given $\pi_a, \pi_b$ in $R(\pi_1, \pi_3)$, we say $\pi_a$ is close behind $\pi_b$ if $0 < \pi_1[\pi_a\cap \pi_1] - \pi_1[\pi_b\cap \pi_1]\le h$ and $0 \le \pi_3[\pi_a\cap \pi_3] - \pi_3[\pi_a\cap \pi_3]\le h$. The point of this definition is that, if $\pi_b$ is pick as the base path, then $(\pi_1, \pi_3, \pi_a\cap \pi_1, \pi_a\cap \pi_3)$ is a tuple in $Q$, in which case we say $\pi_a$ is charged to the pair $(\pi_1, \pi_3)$. 
        
    We first show that if $h\ge 18\frac{\ell}{|R(\pi_1, \pi_3|}$ then the expected number of paths charged to $(\pi_1, \pi_3)$ is at least \begin{equation*}
            \frac{h}{36\ell p}|R(\pi_1, \pi_3)|^2.
        \end{equation*}
        For each $0< i \le z:= |\pi_1|+|\pi_3|+h$, let $a_i$ denotes the number of pair $(\pi_2,j)$ such that $\pi_2\in R(\pi_1, \pi_3)$ and $0 < j\le h$ such that \begin{equation*}
            i = \pi_1[\pi_1\cap \pi_2] + \pi_3[\pi_2\cap \pi_3] + j.
        \end{equation*} Then we have \begin{equation*}
            \sum_{i = 1}^{z}a_i = h |R(\pi_1, \pi_3)|
        \end{equation*}
        since there are $|R(\pi_1, \pi_3)|$ ways to chose $\pi_2$ and $h$ ways to choose $j$. Note that
        $$2z\le 2(4\ell + 4\ell + \ell) = 18 \ell\le h \left|R(\pi_1, \pi_3)\right|.$$
        It follows that \begin{align*}
            \sum_{i = 1}^{z}\binom{a_i}2 &=\frac12\pa{\sum_{i = 1}^{z}a_i ^2 - \sum_{i = 1}^{z}a_i}
            \ge \frac12\pa{\sum_{i = 1}^{z}a_i ^2 - \sum_{i = 1}^{z}a_i}\\
            &\ge \frac12\pa{\frac1z\pa{\sum_{i = 1}^{z}a_i}^2 - \sum_{i = 1}^{z}a_i}& \text{Cauchy -Schwarz}\\
            & = \frac12\pa{\frac{h^2|R(\pi_1, \pi_3)|^2}z - h|R(\pi_1, \pi_3)|}\\
            &\ge \frac{h^2|R(\pi_1, \pi_3)|^2}{4z} \tag*{since $h|R(\pi_1, \pi_3)| \ge 2z$. \qedhere} 
        \end{align*}
    \end{proof}
    Note that $\sum_{i = 1}^{z}\binom{a_i}2$ is the number of unordered pair of distinct tuple $(\pi_a, j_a)$ and $(\pi_b, j_b)$ such that
    $$\pi_1[\pi_1\cap \pi_a] + \pi_3[\pi_a\cap \pi_3] + j_a = \pi_1[\pi_1\cap \pi_b] + \pi_3[\pi_b\cap \pi_3] + j_b.$$
    Call such a pair of tuples an \textit{aligned pair}. It is clear that if $\pi_a = \pi_b$ then this would implies $j_a = j_b$, which contradicts the fact that these are distinct tuples, so we have $\pi_a \ne \pi_b$. Since we count the number of unordered pair of distinct tuple $(\pi_a, j_a)$ and $(\pi_b, j_b)$, we may assume $\pi_b <_R \pi_a$, where $<_R$ is the ordering as in Lemma \ref{lem:rsorder}. Note that then we have \begin{equation*}
        0 < \pi_1[\pi_1\cap \pi_a] - \pi_1[\pi_1\cap \pi_b], 0\le \pi_3[\pi_3\cap \pi_a] - \pi_3[\pi_3\cap \pi_b]
    \end{equation*} and \begin{equation}\label{eq:diffj}
    \pa{\pi_1[\pi_1\cap \pi_a] - \pi_1[\pi_1\cap \pi_b]} + \pa{\pi_3[\pi_3\cap \pi_a] - \pi_3[\pi_3\cap \pi_b]} = j_b - j_a \le h.
    \end{equation}
    so this implies that $\pi_a$ is close behind $\pi_b$. Conversely, if $\pi_a$ is close behind $\pi_b$ then from \eqref{eq:diffj}, the number of way to pick corresponding $j_b, j_a$ such that $(\pi_a, j_a), (\pi_b, j_b)$ forms an aligned pair is at most $h$ since we need $j_b - j_a$ to have a specific value. Thus the number of close behind pairs $\pi_a, \pi_b$ is at least $\frac1h$ the number of aligned pairs, which is at least \begin{equation*}
\frac1h\pa{\frac{h^2|R(\pi_1, \pi_3)|^2}4z} = \frac{h|R(\pi_1, \pi_3|^2}{4z}\ge \frac{h|R(\pi_1, \pi_3)|^2}{36\ell}.
    \end{equation*}
    Thus the expected number of path charged to $(\pi_1, \pi_3)$ when sampling a base path at random is at least \begin{equation*}
\frac{h|R(\pi_1, \pi_3)|^2}{36\ell p}.
    \end{equation*}
    Now the expected total number of elements in $Q$ is at least \begin{align*}
        \sum_{(\pi_1, \pi_3): |R(\pi_1, \pi_3)|\ge \frac{18\ell}{h}}\frac{h|R(\pi_1, \pi_3)|^2}{36\ell p} &= \sum_{\pi_1, \pi_3}\frac{h|R(\pi_1, \pi_3)|^2}{36\ell p} - \sum_{(\pi_1, \pi_3): |R(\pi_1, \pi_3)|< \frac{18\ell}{h}}\frac{h|R(\pi_1, \pi_3)|^2}{36\ell p}\\
        &\ge \sum_{\pi_1, \pi_3}\frac{h|R(\pi_1, \pi_3)|^2}{36\ell p} - p^2\pa{\frac{h(18\ell/h)^2}{36\ell p}}\\
         &\ge \frac{h}{36\ell p}\pa{\sum_{\pi_1, \pi_3}|R(\pi_1, \pi_3)|^2 - 400p^2\ell^2/h^2}
    \end{align*}
    Note that by Cauchy-Schwarz we have \begin{equation*}
        \sum_{\pi_1, \pi_3}|R(\pi_1, \pi_3)|^2\ge \frac1{p^2}\pa{\sum_{\pi_1, \pi_3}|R(\pi_1, \pi_3)|}^2\ge \Omega\pa{\frac1{p^2}\pa{p\ell^2 d^2}^2}\ge \Omega\pa{\ell^4 d^4}
    \end{equation*}
    and if $h\ge \frac{Cp}{\ell d^2}$ we have $400p^2\ell^2/h^2 \le \frac{400}{C^2}\ell^4 d^4$. Thus by choosing $C$ to be large enough, we have \begin{equation*}
        \sum_{\pi_1, \pi_3}|R(\pi_1, \pi_3)|^2 - 400p^2\ell^2/h^2\ge \Omega \pa{\sum_{\pi_1, \pi_3}|R(\pi_1, \pi_3)|^2}
    \end{equation*}
    and thus \begin{equation*}
        \E[|Q|]\ge \Omega\pa{\frac{h}{\ell p}\sum_{\pi_1, \pi_3\in \Pi} \abs*{R(\pi_1, \pi_3)}^2}.
    \end{equation*}
\end{lemma}

We now start proving the upper bound for $|Q|$. Our strategy is as follows. Given a vertex $u$ in the $h$ system, let its \textit{branching degree} be the number of branching path passing through $u$. We will split the vertices into $O(\log n)$ buckets, each with branching degree between $x$ and $2x$ for some $x$. A Cauchy-Schwarz argument would show that, roughly speaking, at least $\frac1{O\pa{\log n}}$ of elements in $Q$ comes from pairs $u, v$ in the same bucket. Note that if $u,v$ is on a same path $\pi_2$ and  have branching degree $\Theta(x)$, this will contributes $x^2$ elements to $Q$. The number of pairs of vertices on the same path is the sum of the square of path length of certain subsystem, so it is helpful to have an upper bound on the sum of the square path length. We would also need upper bounds on the number of vertices of branching degree $x$ or higher, since larger degree vertices contributes more to $Q$. We first start by giving an upper bound for sum of the square of path length.

\begin{lemma}\label{dominatedsq}
    Let $a_1\geq a_2\geq \dots \geq a_n\geq 0$ be a decreasing sequence of nonnegative real numbers and $b_1, \dots, b_n$ is a sequence of nonnegative real numbers such that \[\sum_{i = 1}^k a_i\leq \sum_{i = 1}^k b_i\]
    for any $1\leq k\leq n$, then we have \[\sum_{i = 1}^n a_i^2\leq \sum_{i = 1}^n b_i^2.\]
    \begin{proof}
    Without loss of generality, we may assume $b_1\ge \dots \ge b_n$, (otherwise we sort the $b_i$ descendingly, it is clear that the assumption $\sum_{i = 1}^k a_i\leq \sum_{i = 1}^k b_i$ still hold with the sorted sequence.)
    Let $a_{n+1} = b_{n+1} = 0$. We have 
    \begin{align*}
        \sum_{i = 1}^n b_i^2- \sum_{i = 1}^n a_i^2 &= \sum_{i = 1}^n (b_i - a_i)(b_i+a_i)\\
        &= \sum_{i = 1}^{n}\pa{\sum_{j = 1}^i(b_j - a_j)}\pa{b_i + a_i - b_{i+1} - a_{i+1}}\\
        &\ge 0
    \end{align*}
    where we used summation-by-part in the last equality above.
    \end{proof}
\end{lemma}
\begin{lemma}[Continuous version]
    Let $a_1\geq a_2\geq \dots \geq a_n\geq 0$ be a decreasing sequence of nonnegative real numbers and $f\colon[0, n]\to \R_{\geq 0}$ is a nonnegative function such that \[\sum_{i = 1}^k a_i\leq \int_0^k f\]
    for any $1\leq k\leq n$, then we have \[\sum_{i = 1}^n a_i^2\leq \int_0^n f^2.\]
    \begin{proof}
        Let \begin{equation*}
            b_i := \int_{i-1}^i f
        \end{equation*}
        then we have \[\sum_{i = 1}^k a_i\leq \sum_{i = 1}^k b_i\]
    for any $1\leq k\leq n$, so by Lemma \ref{dominatedsq} we have 
    \begin{equation*}
            \sum_{i = 1}^n a_i^2\leq \sum_{i = 1}^n b_i^2 = \sum_{i = 1}^n \pa{\int_{i-1}^i f}^2\le \sum_{i = 1}^n \int_{i-1}^i f^2 = \int_0^n f^2
        \end{equation*}
    where the inequality $\pa{\int_{i-1}^i f}^2\le \int_{i-1}^i f^2$ follows from Cauchy-Schwarz (applied with $f$ and the constant function $1$.)
    \end{proof}
\end{lemma}

\begin{lemma}\label{lem:l2bound}
    Let $S'$ be a half-$4$-bridge free system on at most $n_1$ vertices and average degree at most $d$, and max length at most $\ell$. Suppose that we have
    $$H(n, p, 4)\le O\left(n + n^{a}p^{2-2a} + n^{2-2b}p^b + p\right)$$
    with $\frac23 \le b < 0.701$ and $\frac{8}{11}-0.001\le a < \frac34$. Let $\norm{S'}_2^2$ denotes the sum of the path length squared.
    Then 
    \[\norm{S'}_2^2\leq O\pa{n_1(\ell+d)  + n_1^{3/2}d^{\frac{3-4a}{2-2a}}+ n_1^{\frac 1b}d^{2-\frac1b}}.\]
    \begin{proof}
        Let $a_1 \ge\dots\ge a_p$ be the length of the paths in the system. We have $\norm{S'} \le O(n_1d)$. Then for each $1\leq k \leq p$ we have \[\sum_{i = 1}^k a_i\leq C\min(k\ell, \max(n_1, n_1^{a}k^{2-2a}, n_1^{2-2b} k^{b}, k), n_1d)\]
        for some absolute constant $C$, since the longest $k$ path forms a half-$4$-bridge free system. Thus let $f$ be the derivative of the above function with respect to $k$, which, through some straightforward calculus (omitted), turns out to be \[f(x) = \begin{cases}
            C\ell, &0\leq x\leq \min\pa{\frac{n_1}\ell, \frac{n_1^{\frac{a}{2a-1}}}{\ell^{\frac1{2a-1}}}},\\
            0, &\min\pa{\frac{n_1}\ell, \frac{n_1^{\frac{a}{2a-1}}}{\ell^{\frac1{2a-1}}}} \leq x \leq \min\pa{\sqrt{n_1}, \frac{n_1^{\frac{a}{2a-1}}}{\ell^{\frac1{2a-1}}}}, \\
           (2-2a)Cn_1^{a}x^{1-2a}, &\min\pa{\sqrt{n_1}, \frac{n_1^{\frac{a}{2a-1}}}{\ell^{\frac1{2a-1}}}}\leq x\leq \min\pa{n_1^{\frac{2b+ a - 2}{2a+b-2}}, n_1^{1/2}d^{\frac1{2-2a}}},\\
            bCn_1^{2-2b}x^{b-1}, & \min\pa{n_1^{\frac{2b+ a - 2}{2a+b-2}}, n_1^{1/2}d^{\frac1{2-2a}}}\leq x\leq \min\pa{n_1^{2-\frac1b}d^{\frac1b}, n_1^{1/2}d^{\frac1{2-2a}}, n_1^2},\\
            1, &\min\pa{n_1^{2-\frac1b}d^{\frac1b}, n_1^{1/2}d^{\frac1{2-2a}}, n_1^2}\le x\le \min\pa{n_1^{2-\frac1b}d^{\frac1b}, n_1^{1/2}d^{\frac1{2-2a}}, n_1d, p},\\
            0,  &\min\pa{n_1^{2-\frac1b}d^{\frac1b}, n_1^{1/2}d^{\frac1{2-2a}}, n_1d, p} \le x\le p,
        \end{cases}\]
        we have \[\sum_{i = 1}^k a_i\leq \int_0^k f.\]
        Thus  \[\sum_{i = 1}^p a_i^2\leq \int_0^p f^2 \le O\pa{n_1(\ell+d)  + n_1^{3/2}d^{\frac{3-4a}{2-2a}}+n_1^{\frac 1b}d^{2-\frac1b}},\]
        where we omitted some straightforward calculus calculations.
    \end{proof}
\end{lemma}

We now start proving bounds on the number of vertices with branching degree $\Omega(x)$.

\begin{lemma}\label{lem:hbackward}
    Fix a vertex $v$. Consider $h$ previous vertices in the each of the $\Theta(d)$ path passing through $v$.
    Suppose that we have
    $$H(n, p, 4)\le O(n + n^{a}p^{2-2a} + n^{2-2b}p^b + p)$$
    with $\frac23 \le b < 0.701$ and $\frac{8}{11}-0.001\le a < \frac34$. Consider the induced system on those $O(dh)$ vertices. Then the number of path of length at least $\Omega(x)$ 
    in the subsystem is at most
    \begin{equation*}
        O\pa{\min\{\frac{dh}x + \frac{(dh)^{\frac{a}{2a-1}}}{x^{\frac1{2a-1}}} + \frac{(dh)^2}{x^{\frac1{1-b}}}, \frac{d^2h}{x}\}}.
    \end{equation*}
    \begin{proof}
        Let $p_x$ be the number of paths. Note that there are $O(dh)$ with degree $O(d)$ so we have $xp_x\le (dh)d$ so \begin{equation*}
            p_x\le O\pa{\frac{d^2h}{x}}.
        \end{equation*} If $x$ is at most a constant we are done since this term is smaller in the min. Otherwise we have 
        \begin{equation*}
            p_x x\le O\pa{dh + (dh)^{a}p_x^{2-2a} + \pa{dh}^{2-2b}p_x^b}.
        \end{equation*}
        Rearranging algebraically gives 
        \begin{equation*}
        p_x\le O\pa{\frac{dh}x + \frac{(dh)^{\frac{a}{2a-1}}}{x^{\frac1{2a-1}}} + \frac{(dh)^2}{x^{\frac1{1-b}}}}. \qedhere
    \end{equation*}
    \end{proof}
\end{lemma}
\begin{lemma}
     Let $n_x$ denotes the number of vertices in the subsystem that has branching degree $\Omega(x)$, not including those on the base path itself. Suppose that we have
     $$H(n, p, 4)\le O\left(n + n^{a}p^{2-2a} + n^{2-2b}p^b + p\right)$$
     with $\frac23 \le b < 0.701$ and $\frac{8}{11}-0.001\le a < \frac34$. Then we have  
    \begin{equation*}
        \E[n_x]\le O\pa{\frac{\ell h}{x} + \frac{\ell d^{\frac{1-a}{2a-1}}h^{\frac{a}{2a-1}}}{x^{\frac1{2a-1}}} + \frac{\ell dh^2}{x^{\frac{1}{1-b}}}},
    \end{equation*}
    and $n_x\le O(\frac{\ell d h}{x})$ deterministically.
    \begin{proof}
        There are $O(\ell d)$ branching paths, so the number of vertices on these path couting with multiplicity is at most $O(\ell d h)$. Thus the number of vertices that is on $\Omega(x)$ of the branching paths is at most $O\pa{\frac{\ell d h}x}$ deterministically. 
        By Lemma \ref{lem:hbackward}, there are at most \begin{equation*}
        O\pa{\frac{ndh}x + \frac{n(dh)^{\frac{a}{2a-1}}}{x^{\frac1{2a-1}}} + \frac{n(dh)^2}{x^{\frac1{1-b}}}}
    \end{equation*}
    pairs of vertex $v$ and path $\pi_b$ such that $v$ is in at least $\Omega(x)$ of the branching path if $\pi_b$ was pick as a base path, and $v\notin \pi_b$. Thus the expected number of vertices which lies in at least $\Omega(x)$ branching path when a random base path is chosen is at most 
    \begin{equation*}
        O\pa{\frac{ndh}{px} + \frac{n(dh)^{\frac{a}{2a-1}}}{px^{\frac1{2a-1}}} + \frac{n(dh)^2}{px^{\frac1{1-b}}}} = O\pa{\frac{\ell h}{x} + \frac{\ell d^{\frac{1-a}{2a-1}}h^{\frac{a}{2a-1}}}{x^{\frac1{2a-1}}} + \frac{\ell dh^2}{x^{\frac1{1-b}}}}. \qedhere
    \end{equation*}
    \end{proof}
\end{lemma}


\subsection{Completing the Proof}

The previous lemma relates the value of $H(n, p, 4)$ to $n_x$, which in turn relates to upper bounds on $Q$.
Since an upper bound on $Q$ implies an upper bound on $H(n, p, 4)$ in turn, it is intuitive that this will imply \emph{some} upper bound on $H(n, p, 4)$.
It requires some straightforward but tedious algebra to realize this upper bound.

We will next start combining the upper bound on the sum of square of path length, and on $n_x$, for upper bounds on $|Q|$. Note that, $n_x$ only counts vertices that are not on the base path, but an element $(\pi_1, \pi_3, u, v)$ could potentially have $v\in \pi_b$ (the base path.) Thus we have to deal with elements $(\pi_1, \pi_3, u, v)$ of $Q$ with $v\in \pi_b$ separately. Define \begin{equation*}
    Q_1 := \{(\pi_1, \pi_3, u, v)\in Q: v\notin \pi_b\}
\end{equation*}
and \begin{equation*}
    Q_2 := \{(\pi_1, \pi_3, u, v)\in Q: v\in \pi_b\}.
\end{equation*}

\begin{lemma}
     Suppose that we have
     $$H(n, p, 4)\le O\left(n + n^{a}p^{2-2a} + n^{2-2b}p^b + p\right)$$
     with $\frac23 \le b < 0.701$ and $\frac{8}{11}-0.001\le a < \frac34$. Suppose $h\le \min(\ell, d)$. Then we have
    \begin{equation*}
        \E[|Q_1|]\le \Tilde{O}\pa{\ell^2 d h^{\frac1b} + \ell^{\frac32}d^{\frac{6-7a}{2-2a}}h^{\frac{2b+1}{2b}} + \ell^{\frac1b}d^2h^{\frac{4b-2b^2-1}{b^2}}}
    \end{equation*}
    \begin{proof}
         
        Firstly, we consider the tuples $(\pi_1, \pi_3, u, v)$ such that $v\notin \pi_b$, let the number of such tuple be $Q_1$.
        For each path $\pi$ and $x$, let $Q_\pi$ denotes the contribution of $\pi$ to $Q_1$, which is the number of tuple $(\pi_1, \pi_3, u, v)$ such that $u <_{\pi}v$ and $\pi_1, \pi_3$ are branching paths passing through $u$ and $v$, and $u, v\notin \pi_b$, and $\pi_3 < \pi_2$. Let $\pi_x$ denotes the number of vertices on $\pi$ with branching degree between $x$ and $2x$, not including vertices in the base graph if any. Then we have 
        \begin{align*}
            Q_1 =  \sum_{\pi}Q_\pi &\le \sum_{\pi}\pa{\sum_{x}x\pi_x}^2\\ 
            &\le O(\log n)\sum_{\pi}\sum_{x}(x\pi_x)^2 \tag*{Cauchy-Schwarz}\\
            &\le O(\log n)\sum_x x^2\sum_{\pi}\pi_x^2\\
            &\le O(\log n)\sum_x x^2 \pa{n_x\ell + n_x^{3/2}d^{\frac{3-4a}{2-2a}} + n_x^{\frac 1b}d^{2-\frac 1b}} \tag*{Lemma \ref{lem:l2bound}}
        \end{align*}
        where in all of the above, the sum over $x$ runs over all powers of $2$ from $x = \Theta(1)$ up to $x = \Theta(d)$. Thus \begin{equation*}
            \E[Q_1]\le \Tilde{O}\pa{\sum_x x^2 \pa{\E[n_x]\ell  + \E[n_x^{3/2}]d^{\frac{3-4a}{2-2a}} + \E[n_x^{\frac1b}]d^{2-\frac1b}}}.
        \end{equation*}
    Since the above sum has $O(\log n)$ term, it suffices to show that for any  $x$ we have 
    \begin{equation*}
        x^2 \pa{\E[n_x](\ell+d)  + \E[n_x^{3/2}]d^{\frac{3-4a}{2-2a}} + \E[n_x^{\frac1b}]d^{2-\frac1b}}\le O\pa{\ell^2 d h^{\frac1b} + \ell^{\frac32}d^{\frac{6-7a}{2-2a}}h^{\frac{2b+1}{2b}} + \ell^{\frac1b}d^2h^{\frac{4b-2b^2-1}{b^2}}}  \end{equation*}
    We consider the possible values of $x$.
        \begin{itemize}
            \item \textbf{Case 1: $x\le O(h^{\frac{1-b}{b}})$.}
            
            Note that $x = h^{\frac{1-b}b}$ is the threshold at which $\frac{\ell d h}x = \frac {\ell d h^2}{x^{\frac{1-b}{b}}}$. We have $  n_x \le O\pa{\frac{\ell d h}x}$ deterministically, so we have 
            \begin{align*}
                &x^2 \pa{\E[n_x](\ell+d)  + \E[n_x^{3/2}]d^{\frac{3-4a}{2-2a}} + \E[n_x^{\frac1b}]d^{2-\frac1b}}\\
                &\le x^2 \pa{\pa{\frac{\ell d h}x}(\ell+d)  + \pa{\frac{\ell d h}x}^{3/2}d^{\frac{3-4a}{2-2a}} + \pa{\frac{\ell d h}x}^{\frac1b}d^{2-\frac1b}} \\
                & \le O\pa{(\ell+d)\ell d h x + (\ell dh)^{3/2}x^{1/2}d^{\frac{3-4a}{2-2a}} + \pa{\ell dh}^{\frac1b} x^{2-\frac1b} d^{2-\frac1b}}\\
                &\le O\pa{\ell^2 d h^{\frac1b} + \ell^{\frac32}d^{\frac{6-7a}{2-2a}}h^{\frac{2b+1}{2b}} + \ell^{\frac1b}d^2h^{\frac{4b-2b^2-1}{b^2}}} \tag*{since $x\le h^{\frac{1-b}{b}}$.}
            \end{align*}
        \item \textbf{Case 2: $\Omega\pa{h^\frac{1-b}b}\le x\le O\pa{(dh)^{\frac{(3a-2)(1-b)}{2a+b-2}}}$.}
        
        In this case, $\E[n_x] \le O\pa{\frac{\ell d h^2}{x^{\frac1{1-b}}}}$. We still have $n_x\le O\pa{\frac{\ell d h}{x}}$ deterministically. Hence
        \begin{equation*}
            \E[n_x^{t}]\le O\pa{\E[n_x]\pa{\frac{\ell d h}{x}}^{t-1}} \le O\pa{\frac{\ell d h^2}{x^{\frac1{1-b}}}\pa{\frac{\ell d h}{x}}^{t-1}} = O\pa{\frac{(\ell d)^t h^{t+1}}{x^{\frac{b}{1-b}+t}}}
        \end{equation*}
        for any constant $t\ge 1$. Thus we have 
            \begin{align*}
                &x^2 \pa{\E[n_x](\ell+d)  + \E[n_x^{3/2}]d^{\frac{3-4a}{2-2a}} + \E[n_x^{\frac1b}]d^{2-\frac1b}}\\
                &\le x^2 \pa{
                \pa{\frac{\ell d h^2}{x^{\frac1{1-b}}}}(\ell+d)  + \pa{\frac{(\ell d)^{3/2} h^{5/2}}{x^{\frac{b}{1-b}+\frac32}}}d^{\frac{3-4a}{2-2a}} + \pa{\frac{(\ell d)^{\frac1b} h^{\frac{b+1}b}}{x^{\frac{b}{1-b}+\frac1b}}}d^{2-\frac1b}} \\
                & \le O\pa{(\ell+d)\ell d h^2 x^{-\frac{2b-1}{1-b}} + (\ell d)^{3/2}h^{5/2}x^{-\frac{3b-1}{2(1-b)}}d^{\frac{3-4a}{2-2a}} + \pa{\ell d}^{\frac1b}h^{\frac{b+1}{b}} x^{-\frac{3b^2-3b+1}{b(1-b)}} d^{2-\frac1b}}\\
                &\le O\pa{\ell^2 d h^{\frac1b} + \ell^{\frac32}d^{\frac{6-7a}{2-2a}}h^{\frac{2b+1}{2b}} + \ell^{\frac1b}d^2h^{\frac{4b-2b^2-1}{b^2}}} \tag*{since $x\ge h^{\frac{1-b}{b}}$.}
            \end{align*}

        \begin{mdframed}
        \textbf{Remark.} It is not a coincidence that the bound for these first two cases is the same.
        In the first case, we obtained an upper bound which turns out to be an increasing function in $x$: namely, the function
        $$f_1(x) := (\ell+d)\ell d h x + (\ell dh)^{3/2}x^{1/2}d^{\frac{3-4a}{2-2a}} + \pa{\ell dh}^{\frac1b} x^{2-\frac1b} d^{2-\frac1b}.$$
        Thus the upper bound is obtained from plugging in the largest value of $x$, which is $h^{\frac{1-b}{b}}$.
        Meanwhile, in the second case, we obtained an upper bound which turns out to be an decreasing function in $x$: namely, the function
        $$f_2(x):= (\ell+d)\ell d h^2 x^{-\frac{2b-1}{1-b}} + (\ell d)^{3/2}h^{5/2}x^{-\frac{3b-1}{2(1-b)}}d^{\frac{3-4a}{2-2a}} + \pa{\ell d}^{\frac1b}h^{\frac{b+1}{b}} x^{-\frac{3b^2-3b+1}{b(1-b)}} d^{2-\frac1b}.$$
        Thus the upper bound is when we plug in the smallest value of $x$ which is $h^{\frac{1-b}{b}}$.
        Notice that these two $f_1, f_2$ functions are obtained by plugging in appropriate upper bound for $\E[n_x], \E[n_x^{3/2}], \E[n_x^{\frac 1b}]$.
        Since $x = h^{\frac{1-b}{b}}$ is the threshold value at which $\frac{\ell d h}x = \frac {\ell d h^2}{x^{\frac{1-b}{b}}}$, the appropriate upper bound for $\E[n_x], \E[n_x^{3/2}], \E[n_x^{\frac1b}]$ in these two cases would be the same when $x = h^{\frac{1-b}{b}}$, which means that $f_1(x) = f_2(x)$ when $x = h^{\frac{1-b}{b}}$, and thus the two upper bounds agree.
\end{mdframed}

        \item \textbf{Case 3: $\Omega\pa{(dh)^{\frac{(3a-2)(1-b)}{2a+b-2}}}\le x \le O\pa{(dh)^{\frac12}}$.}
        
        In this case,
        $$\E[n_x] \le O\pa{\frac{\ell d^{\frac{1-a}{2a-1}} h^{\frac{a}{2a-1}}}{x^{\frac1{2a-1}}}}.$$
        We still have $n_x\le O\pa{\frac{\ell d h}{x}}$ deterministically. Hence
        \begin{equation*}
            \E[n_x^{t}]\le O\pa{\E[n_x]\pa{\frac{\ell d h}{x}}^{t-1}} \le O\pa{\pa{\frac{\ell d^{\frac{1-a}{2a-1}} h^{\frac{a}{2a-1}}}{x^{\frac1{2a-1}}}}\pa{\frac{\ell d h}{x}}^{t-1}}
        \end{equation*}
        We can then give an upper bound for the value of 
        \begin{equation*}
            x^2 \pa{\E[n_x]\ell  + \E[n_x^{3/2}]d^{\frac{3-4a}{2-2a}} + \E[n_x^{\frac1b}]d^{2-\frac1b}}
        \end{equation*}
        by plugging in the upper bound for $\E[n_x^{t}]$ as done in the previous case, and thus would obtain an upper bound as some function $f_3(x)$. To simplify calculations, we will just compute the power of $x$ in the terms here. Note that $E[n_x^t]$ is bounded by a term whose power of $x$ is \begin{equation*}
            -\frac1{2a-1} - (t-1)
        \end{equation*} 
        Then the power of $x$ in the first, second and third term respectively would be \begin{equation*}
            2-\frac1{2a-1} - (t-1)
        \end{equation*} 
        for $t = 1, \frac32, \frac1b$ respectively. Note that \begin{equation*}
            2-\frac1{2a-1} - (t-1)\le 2-\frac1{2a-1} < 0
        \end{equation*} 
        since $a < \frac34$ and $t \ge 1$, so the function $f_3$ is decreasing in $x$. So we have $f_3(x)$ is largest when
        $$x = \Theta\pa{(dh)^{\frac{(3a-2)(1-b)}{2a+b-2}}}.$$
        By similar reasoning to last case, we have that $f_3(x) = f_2(x)$ when
        $$x = \Theta\pa{(dh)^{\frac{(3a-2)(1-b)}{2a+b-2}}}.$$
        And we proved
        \begin{equation*}
            f_2(x) \le f_2(h^{\frac{1-b}{b}})\le O\pa{\ell^2 d h^{\frac1b} + \ell^{\frac32}d^{\frac{6-7a}{2-2a}}h^{\frac{2b+1}{2b}} + \ell^{\frac1b}d^2h^{\frac{4b-2b^2-1}{b^2}}} 
        \end{equation*}
        so we have
        $$f_3(x)\le O\pa{\ell^2 d h^{\frac1b} + \ell^{\frac32}d^{\frac{6-7a}{2-2a}}h^{\frac{2b+1}{2b}} + \ell^{\frac1b}d^2h^{\frac{4b-2b^2-1}{b^2}}}.$$
        
        \item \textbf{Case 4: $\Omega\pa{(dh)^{\frac12}}\le x\le O(d)$.}
        
        In this case we have $\E[n_x]\le O(\frac{\ell h}x)$, and $n_x\le O\pa{\frac{\ell d h}x}$. 
        Thus \begin{equation*}
            \E[n_x^{t}]\le O\pa{\pa{\frac{\ell h}x}\pa{\frac{\ell d h}x}^{t-1}} = O\pa{\ell^t d^{t-1}h^{t}x^{-t}}. 
        \end{equation*}
        Thus 
        \begin{align*}
            &x^2 \pa{\E[n_x](\ell+d)  + \E[n_x^{3/2}]d^{\frac{3-4a}{2-2a}} + \E[n_x^{\frac1b}]d^{2-\frac1b}}\\
            &\le O\pa{(\ell+d)\ell hx + \ell^{3/2}d^{1/2+\frac{3-4a}{2-2a}}h^{3/2}x^{1/2} + \ell^{\frac1b}dh^{\frac1b}x^{2-\frac1b}}\\
            &\le O\pa{(\ell+d)\ell dh + \ell^{3/2}d^{\frac{5-6a}{2-2a}}h^{3/2} + \ell^{\frac1b} d^{3-\frac1b}h^{\frac1b}} \tag*{since $x\le d$}\\
            &\le O\pa{\ell^2 d h^{\frac1b} + \ell^{\frac32}d^{\frac{6-7a}{2-2a}}h^{\frac{2b+1}{2b}} + \ell^{\frac1b}d^2h^{\frac{4b-2b^2-1}{b^2}}}. \qedhere
        \end{align*}
    
        \end{itemize} 
    \end{proof}
\end{lemma}
\begin{lemma}
     Suppose that we have $H(n, p, 4)\le O(n + n^{a}p^{2-2a} + n^{2-2b}p^b + p)$ with $\frac23 \le b < 0.701$ and $\frac{8}{11}-0.001\le a < \frac34$. Suppose $h\le \min(\ell, d)$. Then we have
    \begin{equation*}
        \E[|Q_2|]\le \Tilde{O}\pa{\ell^{3/2}d^2h^{5/4}} + O(\ell^2 d^2)
    \end{equation*}
    \begin{proof}
         
        Let $Q_{2x}$ denotes the set of  tuple in $Q_2$ such that $u$ has branching degree between $x$ and $2x$

    Now, fix a vertex $u$, we show an upper bound on the number of pairs $(\pi_b, v)$ such that the $h$-system with base path $\pi_b$ has $\Theta(x)$ branching path passing through $u$, and $v\in \pi_b$, and $v$ follows $u$ in some path - call such a pair \textit{important}. Let $p_x$ denotes the number of paths with at least $x$ vertices in the $O(dh)$ vertices that precedes $u$ no further than $h$ away in some path. Consider the system $T$ induced by taking the paths to be these $p_x$ path, and taking the vertices to be the union of the $O(d\ell)$ vertices that follow $u$ in some path and the $O(dh)$ vertices that precedes $u$ no further than $h$ away in some path. 
    Note that the number of important pairs $(\pi_b, v)$ is no more than $\norm{T}$. Note that $T$ is source-restricted into the $O(dh)$ vertices preceding $u$. This is because if the source of a path $\pi_b$ was a vertex $v$ that follows $u$ in some path $\pi$ instead of preceding $u$, let $v'$ be one of the $\Theta(x)$ vertex on $\pi_b$ that precede $u$ in some path $\pi'$, then $\pi_b, \pi', \pi$ forms a $3$-cycle. Also, $T$ has at most $p_x$ paths and at most $O(d \ell)$ vertices. Thus by the bound on the path system in the proof of Theorem \ref{thm:onlinesource} we have 
    \begin{equation*}
        \norm{T}\le O\pa{\ell d + \pa{(\ell d) (d h) p_x}^{1/2}} = O\pa{\ell d + \ell^{1/2}dh^{1/2}p_x^{1/2}}
    \end{equation*}
    Recall from Lemma \ref{lem:hbackward} that we have \begin{equation*}
        p_x\le O\pa{\min\left\{\frac{dh}x + \frac{(dh)^{\frac{a}{2a-1}}}{x^{\frac1{2a-1}}} + \frac{(dh)^2}{x^{\frac1{1-b}}}, \frac{d^2h}{x} \right\}}
    \end{equation*}
    and we now use a crude bound
    \begin{equation*}
        p_x\le O\pa{\min\left\{\frac{dh}x + \frac{(dh)^{\frac{a}{2a-1}}}{x^{\frac1{2a-1}}} + \frac{(dh)^2}{x^{\frac1{1-b}}}, \frac{d^2h}{x}\right\}} \le O\pa{\min\left\{\frac{dh}x + \frac{(dh)^2}{x^3}, \frac{d^2h}{x}\right\}}
    \end{equation*}
    Thus 
    \begin{equation*}
        \norm{T}\le O\pa{\ell d + \min\left\{\ell^{1/2}d^{3/2}hx^{-1/2} + \ell^{1/2}d^{2}h^{3/2}x^{-3/2}, \ell^{1/2}d^{2}hx^{-1/2}\right\}}
    \end{equation*}
    We then have that the number of triples $(u, \pi_b, v)$ such that $u$ has branching degree $\Theta(x)$ when $\pi_b$ is selected as the base path and $v\in \pi_b$ and $v$ follows $u$ in some path is at most $n$ times the above bound, and the note that such a triple $(u, \pi_b, v)$ contributes $\Theta(xd)$ to the size of $Q_{2x}$ with probability $\frac 1p$. Thus we have 
    \begin{align*}
        \E[|Q_{2x}|]&\le  O\pa{\frac {nxd}p\pa{\ell d + \min\left\{\ell^{1/2}d^{3/2}hx^{-1/2} + \ell^{1/2}d^{2}h^{3/2}x^{-3/2}, \ell^{1/2}d^{2}hx^{-1/2}\right\}}}\\
        &\le O\pa{\ell^2 d x + \min\left\{\ell^{3/2}d^{3/2}hx^{1/2} + \ell^{3/2}d^{2}h^{3/2}x^{-1/2}, \ell^{3/2}d^{2}hx^{1/2}\right\}}.
    \end{align*}
    Note that
    $$\min\left\{\ell^{3/2}d^{3/2}hx^{1/2} + \ell^{3/2}d^{2}h^{3/2}x^{-1/2}, \ell^{3/2}d^{2}hx^{1/2}\right\} \le \ell^{3/2}d^2 h^{5/4}$$
    for all $x\le d$, by some omitted casework.
    Thus we have 
    \begin{equation*}
        \E[|Q_{2}|]\le \sum_{x}O(\ell^2 d x) + \sum_{x}O\pa{\ell^{3/2}d^2 h^{5/4}} = O\pa{\ell^2 d^2} + \Tilde{O}\pa{\ell^{3/2}d^2 h^{5/4}}.
    \end{equation*}
    where the sum runs over all powers of $2$, and the $\sum_{x}O(\ell^2 d x)$ term is a geometric series that is dominated by the largest term, which is when $x = \Theta(d)$.
    \end{proof}
\end{lemma}

\begin{lemma}\label{lem:qupperbound}
     Suppose that we have $H(n, p, 4)\le O(n + n^{a}p^{2-2a} + n^{2-2b}p^b + p)$ with $\frac23 \le b < 0.701$ and $\frac{8}{11}-0.001\le a < \frac34$. Suppose $h\le \min(\ell, d)$. Then we have 
    \begin{equation*}
        \E[|Q|] \le \Tilde{O}\pa{\ell^2 d h^{\frac1b} + \ell^{\frac32}d^{\frac{6-7a}{2-2a}}h^{\frac{2b+1}{2b}} + \ell^{\frac1b}d^2h^{\frac{4b-2b^2-1}{b^2}}+ \ell^{3/2}d^2h^{5/4}} + O(\ell^2 d^2).
    \end{equation*}
    \begin{proof} We have
        \begin{align*}
        &\E[|Q|]  = \E[|Q_1|]+ \E[|Q_2|]\\
        &\le  \Tilde{O}\pa{\ell^2 d h^{\frac1b} + \ell^{\frac32}d^{\frac{6-7a}{2-2a}}h^{\frac{2b+1}{2b}} + \ell^{\frac1b}d^2h^{\frac{4b-2b^2-1}{b^2}}} + O(\ell^2 d^2) + \Tilde{O}\pa{\ell^{3/2}d^2h^{5/4}}\\
         &= \Tilde{O}\pa{\ell^2 d h^{\frac1b} + \ell^{\frac32}d^{\frac{6-7a}{2-2a}}h^{\frac{2b+1}{2b}} + \ell^{\frac1b}d^2h^{\frac{4b-2b^2-1}{b^2}}+ \ell^{3/2}d^2h^{5/4}} + O(\ell^2 d^2). \qedhere
    \end{align*}
    \end{proof}
\end{lemma}

\begin{lemma}\label{lem:recstep}
     Suppose that we have
     $$H(n, p, 4)\le O(n + n^{a}p^{2-2a} + n^{2-2b}p^b + p)$$
     with $\frac23 \le b < 0.701$ and $\frac{8}{11}-0.001\le a < \frac34$. Then 
    \begin{equation*}
        H(n, p, 4) \le \Tilde{O}\pa{n^{\frac{2+b}{3+b}}p^{\frac{2}{3+b}} + n^{\frac{8b-4b^2-2}{11b-4b^2-3}}p^{\frac{7b-2b^2-2}{11b-4b^2-3}} + \ell^{3/2}d^2h^{5/4}} + O(n + p).
    \end{equation*}
    \begin{proof}
        We set $h  = \frac{Cp}{\ell d^2}$ where $C$ is a large enough constant. We can assume $h\le \min(\ell, d)$, since otherwise through some simple calculations (omitted) we immediately get $\norm{Z}\le O\pa{n^{3/4}p^{1/2}+n^{1/2}p^{3/4} }$ which is better than the above bound.  Then we have 
        \begin{align*}
            &\Tilde{O}\pa{\ell^2 d h^{\frac1b} + \ell^{\frac32}d^{\frac{6-7a}{2-2a}}h^{\frac{2b+1}{2b}} + \ell^{\frac1b}d^2h^{\frac{4b-2b^2-1}{b^2}}+ \ell^{3/2}d^2h^{5/4}} + O(\ell^2 d^2)\\
            &\ge \E[Q] & \tag*{Lemma \ref{lem:qupperbound}}\\
            &\ge \Omega\pa{\frac{h}{\ell p}\sum_{\pi_1, \pi_3\in \Pi} \abs*{R(\pi_1, \pi_3)}^2}& \tag*{Lemma \ref{lem:lowerboundq}} \\
            & \ge \Omega\pa{\frac{h}{\ell p^3}\pa{\sum_{\pi_1, \pi_3\in \Pi} \abs*{R(\pi_1, \pi_3)}}^2} \tag*{Cauchy-Schwarz} \\
            &\ge \Omega\pa{\frac{h}{\ell p^3}\pa{p\ell^2 d^2}^2} \tag*{Lemma \ref{lem:rssize}}\\
            & = \Omega(\ell^2 d^2).
        \end{align*}
        By choosing $C$ to be large enough, the constant inside $\Omega(\ell^2 d^2)$ is larger than the constant inside $O(\ell^2 d^2)$, so we must have 
    \begin{equation*}
        \Tilde{O}\pa{\ell^2 d h^{\frac1b} + \ell^{\frac32}d^{\frac{6-7a}{2-2a}}h^{\frac{2b+1}{2b}+ \ell^{3/2}d^2h^{5/4}} + \ell^{\frac1b}d^2h^{\frac{4b-2b^2-1}{b^2}}+ \ell^{3/2}d^2h^{5/4}}\ge \Omega\pa{\ell^2 d^2}.
    \end{equation*}
    Rearranging, and applying the identity $\norm{Z} = nd 
 = p\ell$, we get 
    \begin{equation*}
        \norm{Z}\le \Tilde{O}\pa{n^{\frac{2+b}{3+b}}p^{\frac{2}{3+b}} + n^{\frac{8b-4b^2-2}{11b-4b^2-3}}p^{\frac{7b-2b^2-2}{11b-4b^2-3}}}.
    \end{equation*}
    (Note: this bound removes some terms that do not end up dominating the sum.)
    \end{proof}
    
\end{lemma}
\begin{theorem}
    Let $\alpha \approx 0.7009$ be a root of $ 4 x^3- 13 x^2+10 x -2$. Then
    \begin{equation*}
        H(n, p, 4) \le O\pa{n + n^{\frac{2+\alpha}{3+\alpha}+o(1)}p^{\frac{2}{3+\alpha}} + n^{2-2\alpha + o(1)}p^{\alpha} + p}.
    \end{equation*}
    Consider the sequences $a_0 = \frac{8}{11}, b_0 = \frac 23$, $a_{i+1} = g(b_i)$ and $b_{i+1} = f(b_i)$ where 
    \begin{equation*}
        g(x) = \frac{2+x}{3+x}, f(x):= \frac{7x-x^2-2}{11x-4x^2-3}
    \end{equation*}
    Note that $\frac 8{11}\le a_0 < \frac34$ and $\frac 23\le b_0 \le \alpha$. Since $f$ is increasing and $f(x) > x$ for $\frac 23\le b_0 \le \alpha$ and $f(\alpha) = \alpha$, we have that $\frac23 < b_i<  \alpha$ for all $i\ge 1$ and $b_i$ is an increasing sequence that converges to $\alpha$. It follows that $\frac 8{11}\le a_i < \frac34$ for all $a_i$.
    Lemma \ref{lem:initialbound} shows that we have 
    \begin{equation*}
        H(n, p, 4)\le O\pa{n + n^{a_0}p^{2-2a_0} + n^{2-2b_0}p^{b_0}+p}
    \end{equation*}
    and so Lemma \ref{lem:recstep} shows that we have \begin{equation*}
        H(n, p, 4)\le \Tilde{O}\pa{ n^{a_1}p^{2-2a_1} + n^{2-2b_1}p^{b_1}} + O(n+p)
    \end{equation*}
    so for any $\varepsilon > 0$ we have 
    \begin{equation*}
        H(n, p, 4)\le O\pa{ n^{a_1-\varepsilon}p^{2-2a_1+2\varepsilon} + n^{2-2b_1+2\varepsilon}p^{b_1-\varepsilon}+n+p}
    \end{equation*}
    We will prove by induction on $i$ that for any $\varepsilon > 0$ we have 
    \begin{equation*}
        H(n, p, 4)\le O\pa{ n^{a_i-\varepsilon}p^{2-2a_i+2\varepsilon} + n^{2-2b_i+2\varepsilon}p^{b_i-\varepsilon}+n+p}
    \end{equation*}
    Above we showed the claim is true for $i = 1$. Suppose the claim is true for $i$, we show it's true for $i +1$. Fix $\varepsilon > 0$. Note that $b_{i+1} = f(b_i)$ and $f$ is continuous so there is $\delta > 0$ such that $f(x) > b_{i+1}-\varepsilon$ for $b_i - \delta < x < b_i$ and $g$ is continuous so there is $\delta_1$ such that $g(x) > a_{i+1} - \varepsilon$ for $b_i-\delta_1 < x < b_i$. Take $\delta' = \frac12\min\{\delta, \delta_1, 0.001, b_i - \frac23\}$ we have \begin{equation*}
        H(n, p, 4)\le O\pa{ n^{a_i-\delta'}p^{2-2a_i+2\delta'} + n^{2-2b_i+2\delta'}p^{b_i-\delta'}+n+p}
    \end{equation*}
    where $\frac 8{11} - 0.001 < a_i - \delta' < \frac 34$ and $\frac 23 < b_i - \delta' < \alpha$. Thus by Lemma \ref{lem:recstep} we have \begin{equation*}
        H(n, p, 4)\le O\pa{ n^{a_{i+1}'}p^{2-2a_{i+1}'} + n^{2-2b_{i+1}'}p^{b_{i+1}'} + n+p}
    \end{equation*}
    where $a_{i+1}' = g(b_i-\delta') > a_{i+1}-\varepsilon$ and $b_{i+1}' = f(b_i-\delta') > b_{i+1}-\varepsilon$, and so the claim is true for $i+1$.
    Finally, it suffices to show that for any $\varepsilon > 0$ we have \begin{equation*}
        H(n, p, 4) \le O\pa{n + n^{\frac{2+\alpha}{3+\alpha}+\varepsilon}p^{\frac{2}{3+\alpha} - 2\varepsilon} + n^{2-2\alpha + 2\varepsilon}p^{\alpha-\varepsilon} + p}.
    \end{equation*}
    Fix $\varepsilon >  0$. Note that $b_i$ converges to $\alpha$ so $a_i$ converges to $\frac{2+\alpha}{3+\alpha}$ so there is $i$ such that $a_i > \frac{2+\alpha}{3+\alpha}+\varepsilon/2$ and $b_i > \alpha - \varepsilon/2$. Then by the above claim we have
    \begin{align*}
        H(n, p, 4) &\le O\pa{ n^{a_i-\varepsilon/2}p^{2-2a_i+\varepsilon} + n^{2-2b_i+\varepsilon}p^{b_i-\varepsilon/2}+n+p}\\
        &\le O\pa{n + n^{\frac{2+\alpha}{3+\alpha}-\varepsilon}p^{\frac{2}{3+\alpha} + 2\varepsilon} + n^{2-2\alpha + 2\varepsilon}p^{\alpha-\varepsilon} + p}
    \end{align*}
    as desired.
\end{theorem}

\section{Offline Pairwise Reachability Preservers \label{app:offlinerp}}

We next prove Theorem \ref{thm:offlinerp}.
See Appendix \ref{app:overview} for a high-level view of the changes from \cite{BHT23} implicit in the following proof.

As in \cite{BHT23}, we use $\beta(n, p, 4)$ to denote the maximum size of a path system that is $4$-bridge-free, where we say a path system is $k$-bridge-free if it has no bridge of size at most $k$.
Furthermore, as shown in \cite{BHT23} (via the Independence Lemma, overviewed in Section \ref{sec:onlinepairwise} that in order to prove Theorem \ref{thm:offlinerp}, it suffices to prove that

$$\beta(n, p, 4)\le O\pa{n + n^{3/4}p^{1/2} + n^{2-\sqrt2 + o(1)}p^{\frac1{\sqrt2}} + p}.$$

We will use the same strategy as for bounding $H(n, p, 4)$.
First note that a similar cleaning lemma hold for bridge-free system.

\begin{lemma}[Cleaning Lemma for bridge-free systems -- c.f. \cite{BHT23}, Lemma 10]
    There exists a $k$-bridge free system on $\le n$ vertices, $\le p$ path whose size is $\Theta(H(n, p, k))$ such that every vertex has degree between $d/4$ and $4d$ and every path has length between $\ell/4$ and $4\ell$ where $d, \ell$ are the average degree and average path length respectively.
\end{lemma}

The proof is almost identical to Lemma \ref{lem:cleaning}, except that we argue that the modified system remains bridge-free instead of half-bridge free, so we omit the proof. 

We shall follow the same strategy with bounding $H(n, p, 4)$. That is, we also use recursion and select the random $h$-system. Note that Lemma \ref{lem:rsorder}, Lemma \ref{lem:rssize}, Lemma \ref{lem:initialbound} and Lemma \ref{lem:lowerboundq} holds in this setting as well, since it holds in a more general setting. In fact the total ordering of Lemma \ref{lem:rsorder} has the following property: If $\pi_a <_R \pi_b$ then $\pi_a \cap \pi_1 <_{\pi_1} \pi_b \cap \pi_1$ and $\pi_a \cap \pi_3 <_{\pi_3} \pi_b \cap \pi_3$. This is because if $\pi_a \cap \pi_3 = \pi_b \cap \pi_3$, then $\pi_1, \pi_b, \pi_a$ will form a $3$ bridge.

As a result, when bounding $|Q|$, we no longer need to deal with the case an element $(\pi_1, \pi_3, u, v)$ of $Q$ has $v$ in the base path. Also note that due to $3$-bridge free, all nodes have branching degree $1$.

We have the following version of Lemma \ref{lem:l2bound}.

\begin{lemma}\label{lem:l2boundbf}
    Let $S'$ be a $4$-bridge free system on at most $n_1 = \Theta(\ell d h)$ vertices and average degree at most $d$, and max length at most $\ell$. Suppose that we have
    $$\beta(n, p, 4)\le O\left(n + n^{3/4}p^{2-2a} + n^{2-2b}p^b + p\right)$$
    with $\frac23 \le b < \frac 34$.
    Then 
    \[\norm{S'}_2^2\leq O\pa{n_1\ell  + n_1^{3/2}\log(n_1)+ n_1^{\frac 1b}d^{2-\frac1b}}\]
    \begin{proof}
        The proof is similar to Lemma \ref{lem:l2bound}, so we will sketch it here.
        Let $a_1 \ge\dots\ge a_p$ be the length of the paths in the system, and so we have
        
        $$\sum_{i = 1}^k a_i\leq C\min(k\ell, \max(n_1, n_1^{3/4}k^{1/2}, n_1^{2-2b} k^{b}, k), n_1d)$$
        for some absolute constant $C$.
        Letting $f$ be the derivative of the above function with respect to $k$, i.e \[f(x) = \begin{cases}
            C\ell, &0\leq x\leq \min\pa{\frac{n_1}\ell, \frac{n_1^{3/2}}{\ell^{2}}},\\
            0, &\min\pa{\frac{n_1}\ell, \frac{n_1^{3/2}}{\ell^{2}}} x \leq \min\pa{\sqrt{n_1}, \frac{n_1^{3/2}}{\ell^{2}}}, \\
           \frac12Cn_1^{3/4}x^{-1/2}, &\min\pa{\sqrt{n_1}, \frac{n_1^{3/2}}{\ell^{2}}}\leq x\leq \min\pa{n_1^{\frac{8b-5}{4b-2}}, n_1^{1/2}d^{2}},\\
            bCn_1^{2-2b}x^{b-1}, & \min\pa{n_1^{\frac{8b-5}{4b-2}}, n_1^{1/2}d^{2}}\le x\le \min\pa{n_1^{2-\frac1b}d^{\frac1b}, n_1^{1/2}d^{2}},\\
            0,  &\min\pa{n_1^{2-\frac1b}d^{\frac1b}, n_1^{1/2}d^{2}} \le x\le p,
        \end{cases}\]
        we have \[\sum_{i = 1}^k a_i\leq \int_0^k f\].
        Thus we have
        \begin{align*}
        \sum_{i = 1}^p a_i^2 &\leq \int_0^p f^2\\
        &\le O\pa{n_1\ell  + n_1^{3/2}\log n_1+n_1^{\frac 1b}d^{2-\frac1b}}. \tag*{\qedhere}
        \end{align*}
    \end{proof}
\end{lemma}
Note that in this case we directly have an upper bound for $|Q|$.
\begin{lemma}
    Suppose that we have
    $$\beta(n, p, 4)\le O\left(n + n^{3/4}p^{2-2a} + n^{2-2b}p^b + p\right)$$
    with $\frac23 \le b < \frac 34$. Then we have
    \begin{equation*}
        |Q|\le O\left(\ell^2 d h + (\ell d h)^{3/2}\log n + \ell^{\frac 1b}d^2 h^{\frac1b}\right)
    \end{equation*}
    \begin{proof}
        Each pair $(u, v)$ in the same path in the subsystem contributes only one element to $Q$, since they only have branching degree $1$. Thus $|Q|$ is less than the sum of square of path length of the randomly sampled $h$-system, which has at most $n_1 = \Theta(\ell dh)$ vertices. Plugging in $n_1 = \Theta(\ell d h)$ we have \begin{equation*}
        |Q|\le O\left(\ell^2 d h + (\ell d h)^{3/2}\log n + \ell^{\frac 1b}d^2 h^{\frac1b} \right). \tag*{\qedhere}
    \end{equation*}
\end{proof}

\end{lemma}
\begin{lemma}\label{lem:recbf}
    Suppose that we have
    $$\beta(n, p, 4)\le O\left(n + n^{3/4}p^{2-2a} + n^{2-2b}p^b + p\right)$$
    with $\frac23 \le b < \frac 34$. Then we have
    $$\beta(n, p, 4)\le O\left(n + n^{3/4}p^{1/2} + n^{\frac{2}{b+1}}p^{\frac{2b+1}{2b+2}} + p\right).$$
    \begin{proof}
        Similar to the non-adaptive setting, we set $h  = \frac{Cp}{\ell d^2}$ where $C$ is a large enough constant. We can assume $h\le \min(\ell, d)$, since otherwise by some straightforward calculations, we immediately get
        $$\norm{Z}\le O\left(n^{1/2}p^{3/4}+ n^{3/4}p^{1/2}\right)$$
        which is better than the above bound. Then we have 
        \begin{align*}
            &O\left(\ell^2 d h + (\ell d h)^{3/2}\log n + \ell^{\frac 1b}d^2 h^{\frac1b}\right)\\
            &\ge \E[|Q|] & \text{Lemma \ref{lem:l2boundbf}}\\
            &\ge \Omega\pa{\frac{h}{\ell p}\sum_{\pi_1, \pi_3\in \Pi} \abs*{R(\pi_1, \pi_3)}^2}& \text{Lemma \ref{lem:lowerboundq}} \\
            & \ge \Omega\pa{\frac{h}{\ell p^3}\pa{\sum_{\pi_1, \pi_3\in \Pi} \abs*{R(\pi_1, \pi_3)}}^2} & \text{Cauchy-Schwarz} \\
            &\ge \Omega\pa{\frac{h}{\ell p^3}\pa{p\ell^2 d^2}^2} &\text{Lemma \ref{lem:rssize}}\\
            & = \Omega(\ell^2 d^2).
        \end{align*}
    Rearranging, and using the identity $\norm{Z} = nd = p\ell $, we get  
    \begin{equation*}
        \norm{Z}\le O\pa{n + n^{3/4}p^{1/2} + n^{\frac{2}{b+1}}p^{\frac{2b+1}{2b+2}} + p}
    \end{equation*}
    (noting again that some terms would end up never dominating the sum and thus have been removed).
    \end{proof}
\begin{theorem}
    We have $\beta(n, p, 4)\le O\pa{n + n^{3/4}p^{1/2} + n^{2-\sqrt2 + o(1)}p^{\frac1{\sqrt2}} + p}$.
\end{theorem}
    \begin{proof}
        Note that by Lemma \ref{lem:initialbound} we have $\beta(n, p, 4)\le O\pa{n + n^{2/3}p^{2/3} + p}$. Let $b_0 = \frac 23$ and $b_{i+1} = \frac{2b_i+1}{2b_i+2}$. Then we have \begin{equation*}
            \beta(n, p, 4)\le O\pa{n + n^{3/4}p^{1/2+} n^{2-2b_0}p^{b_0} + p}
        \end{equation*}
        and using Lemma \ref{lem:recbf}, by induction we have 
        \begin{equation*}
            \beta(n, p, 4)\le O\pa{n + n^{3/4}p^{1/2+} n^{2-2b_i}p^{b_i} + p}
        \end{equation*}
        for all $i$. Since $b_i$ converges to $\frac1{\sqrt 2}$ we have $\beta(n, p, 4)\le O\pa{n + n^{3/4}p^{1/2} + n^{2-\sqrt2 + o(1)}p^{\frac1{\sqrt2}} + p}$.
    \end{proof}
\end{lemma}

\end{document}